\documentclass[a4paper,USenglish,cleveref, numberwithinsect]{lipics-v2019}
\nolinenumbers 
\hideLIPIcs

\usepackage{xspace}
\usepackage{tikz}
\usetikzlibrary{arrows,trees,positioning,calc}
\usetikzlibrary{automata,shapes.arrows,backgrounds,fit}
\usetikzlibrary{decorations.pathmorphing,decorations.pathreplacing}
\usetikzlibrary{shapes}
\tikzstyle{rstate}=[state,ellipse]
\tikzset{>={latex}}

\usepackage[algo2e,boxruled,vlined,linesnumbered]{algorithm2e}
\usepackage{algorithm}
\usepackage{algpseudocode}
\usepackage{stmaryrd}
\usepackage{upgreek}

\title{Dynamic Complexity of Document Spanners} 

\author{Dominik D. Freydenberger}{Loughborough University, Loughborough, United Kingdom}{}{https://orcid.org/0000-0001-5088-0067}{}
\author{Sam M. Thompson}{Loughborough University, Loughborough, United Kingdom}{}{https://orcid.org/0000-0002-3476-6739}{}

\authorrunning{D.\,D.\ Freydenberger and S.\,M.\ Thompson}

\Copyright{Dominik D. Freydenberger and Sam M. Thompson}

\ccsdesc[100]{Theory of computation~Complexity theory and logic}
\ccsdesc[100]{Information systems~Information extraction}

\keywords{Document spanners, information extraction, dynamic complexity, descriptive complexity, word equations}

\category{}

\supplement{}

\acknowledgements{The authors would like to thank the anonymous reviewers as well as Thomas Schwentick for their helpful comments and suggestions. The authors would also like to thank Thomas Zeume for clarifying a result from his thesis.}

\EventEditors{Carsten Lutz and Jean Christoph Jung}
\EventNoEds{2}
\EventLongTitle{23rd International Conference on Database Theory (ICDT 2020)}
\EventShortTitle{ICDT 2020}
\EventAcronym{ICDT}
\EventYear{2020}
\EventDate{March 30--April 2, 2020}
\EventLocation{Copenhagen, Denmark}
\EventLogo{}
\SeriesVolume{155}
\ArticleNo{20}

\newcommand{\erasing}[1]{#1_{\mathsf{E},\Sigma}}
\newcommand{\nonerasing}[1]{#1_{\mathsf{NE},\Sigma}}
\newcommand{\powerset}[1]{\mathcal{P}(#1)}
\newcommand{\df}{:=}
\newcommand{\select}{\upxi}
\newcommand{\join}{\bowtie}
\newcommand{\subst}[3]{#1_{\frac{#2}{#3}}}

\newcommand{\VAset}{\mathsf{VA_{set}}}

\newcommand{\core}{\mathsf{core}}
\newcommand{\spanreg}{\mathsf{reg}}

\newcommand{\diff}{\setminus}
\newcommand{\cored}{\mathsf{core}\cup\{\diff\}}

\newcommand{\rgx}{\mathsf{RGX}}
\newcommand{\RGXreg}{\rgx^{\spanreg}}
\newcommand{\VAsetreg}{\mathsf{VA}_{\mathsf{set}}^{\spanreg}}
\newcommand{\RGXcore}{\rgx^{\core}}
\newcommand{\RGXcored}{\rgx^{\cored}}

\newcommand{\VAsetcore}{\mathsf{VA}_{\mathsf{set}}^{\core}}
\newcommand{\VAsetcored}{\mathsf{VA}_{\mathsf{set}}^{\cored}}

\newcommand{\ECrtext}{\textsf{EC\textsuperscript{reg}}}
\newcommand{\EC}{\ifmmode{\mathsf{EC}}\else{\textsf{EC}}\xspace\fi}
\newcommand{\ECr}{\ifmmode{\mathsf{EC^{reg}}}\else{\ECrtext}\xspace\fi}

\newcommand{\deltah}{\delta^*}
\newcommand{\quotproj}{\pi}

\newcommand{\subword}{\sqsubseteq}

\newcommand{\lang}{\mathcal{L}} 
\newcommand{\rlang}{\mathcal{R}}  
\newcommand{\weqeq}{\mathbin{\dot{=}}}

\newcommand{\spn}[1]{[#1\rangle}

\newcommand{\spanner}[1]{\llbracket #1 \rrbracket}
\newcommand{\SVars}[1]{{\mathsf{SVars}\left(#1\right)}}

\newcommand{\fvar}{\mathsf{free}}
\newcommand{\var}{\mathsf{var}}
\newcommand{\mv}{\mathsf{W}}
\newcommand{\splog}{\ifmmode{\mathsf{SpLog}}\else{\textsf{SpLog}}\xspace\fi}
\newcommand{\splogneg}{\ifmmode{\mathsf{SpLog}^{\neg}}\else{\textsf{SpLog}$^{\neg}$}\xspace\fi}
\newcommand{\dpcsplog}{\ifmmode{\mathsf{DPC}}\else{\textsf{DPC}}\xspace\fi}
\newcommand{\pcsplog}{\ifmmode{\mathsf{PC}}\else{\textsf{PC}}\xspace\fi}
\newcommand{\splogrx}{\ifmmode{\mathsf{SpLog_{rx}}}\else{\textsf{SpLog\textsubscript{rx}}}\xspace\fi}
\newcommand{\splognegrx}{\ifmmode{\mathsf{SpLog_{rx}^{\neg}}}\else{\textsf{SpLog$_{\mathsf{rx}}^{\neg}$}}\xspace\fi}
\newcommand{\dpcsplogrx}{\ifmmode{\mathsf{DPC{rx}}}\else{\textsf{DPC\textsubscript{rx}}}\xspace\fi}
\newcommand{\pcsplogrx}{\ifmmode{\mathsf{PC_{rx}}}\else{\textsf{PC\textsubscript{rx}}}\xspace\fi}
\newcommand{\constr}{\ensuremath{\mathsf{C}}}

\newcommand{\openvar}[1]{\mathbin{\vdash_{#1}}}
\newcommand{\closevar}[1]{\mathbin{\dashv_{#1}}}

\newcommand{\validr}[1]{\mathsf{Ref}(#1)}

\newcommand{\clr}{\mathsf{clr}}

\newcommand{\logeq}{\mathbin{\dot{=}}}
\newcommand{\wordstruc}{\mathcal{W}}
\newcommand{\auxstruc}{\mathcal{W}_{aux}}
\newcommand{\programstate}{\mathcal{S}}
\newcommand{\dynword}{\mathsf{word}}
\newcommand{\worddomain}{D}

\newcommand{\openspanvar}[1]{ #1^{o} }
\newcommand{\closespanvar}[1]{ #1^{c} }
\newcommand{\symbolrel}[1]{R_{#1} }
\newcommand{\updateprogram}{\vec{P}}
\newcommand{\updateformula}[3]{\phi^{#1}_{#2} ( #3 )}
\newcommand{\ins}[2]{\mathsf{ins}_{#1}(#2)}
\newcommand{\reset}[1]{\mathsf{reset}(#1)}

\newcommand{\absins}[1]{\ifmmode{\mathsf{ins_{#1}}}\else{\textsf{ins_{#1}}}\xspace\fi}
\newcommand{\absreset}{\ifmmode{\mathsf{reset}}\else{\textsf{reset}}\xspace\fi}

\newcommand{\unknownupdate}{\partial}
\newcommand{\nextsym}{\leadsto_{w}}
\newcommand{\newnextsym}{\leadsto_{w'}}

\newcommand{\nextrelation}{R_{\mathsf{Next}}}

\newcommand{\lastrel}{R_{\mathsf{last}}}
\newcommand{\firstrel}{R_{\mathsf{first}}}
\newcommand{\position}[2]{\mathsf{pos}_{#1}(#2)}
\newcommand{\equalsubstr}{R_{\mathsf{eq}}}
\newcommand{\symel}[1]{\bigvee\limits_{\xi \in \Sigma}(\symbolrel{\xi}(#1))}
\newcommand{\oneopen}{x_o}
\newcommand{\oneclose}{x_c}
\newcommand{\twoopen}{y_o}
\newcommand{\twoclose}{y_c}
\newcommand{\patternquery}{\mathcal{P}}

\newcommand{\true}{\ifmmode{\mathsf{True}}\else{\textsf{True}}\xspace\fi}
\newcommand{\false}{\ifmmode{\mathsf{False}}\else{\textsf{False}}\xspace\fi}

\newcommand{\cq}{\ifmmode{\mathsf{CQ}}\else{\textsf{CQ}}\xspace\fi}
\newcommand{\ucq}{\ifmmode{\mathsf{UCQ}}\else{\textsf{UCQ}}\xspace\fi}
\newcommand{\dynfo}{\ifmmode{\mathsf{DynFO}}\else{\textsf{DynFO}}\xspace\fi}
\newcommand{\dynqf}{\ifmmode{\mathsf{DynQF}}\else{\textsf{DynQF}}\xspace\fi}
\newcommand{\dynprop}{\ifmmode{\mathsf{DynPROP}}\else{\textsf{DynPROP}}\xspace\fi}
\newcommand{\dyncq}{\ifmmode{\mathsf{DynCQ}}\else{\textsf{DynCQ}}\xspace\fi}
\newcommand{\dynucq}{\ifmmode{\mathsf{DynUCQ}}\else{\textsf{DynUCQ}}\xspace\fi}

\newcommand{\cqwithpre}{\ifmmode{\mathsf{CQ_{pre}}}\else{$\mathsf{CQ_{pre}}$}\xspace\fi}
\newcommand{\ucqwithpre}{\ifmmode{\mathsf{UCQ_{pre}}}\else{$\mathsf{UCQ_{pre}}$}\xspace\fi}
\newcommand{\dynfowithpre}{\ifmmode{\mathsf{DynFO_{pre}}}\else{$\mathsf{DynFO_{pre}}$}\xspace\fi}
\newcommand{\dynpropwithpre}{\ifmmode{\mathsf{DynPROP_{pre}}}\else{$\mathsf{DynPROP_{pre}}$}\xspace\fi}
\newcommand{\dyncqwithpre}{\ifmmode{\mathsf{DynCQ_{pre}}}\else{$\mathsf{DynCQ_{pre}}$}\xspace\fi}
\newcommand{\dynucqwithpre}{\ifmmode{\mathsf{DynUCQ_{pre}}}\else{$\mathsf{DynUCQ_{pre}}$}\xspace\fi}

\newcommand{\emptyword}{\varepsilon}

\newcommand{\nextsubform}[1]{\varphi_{#1}^{\nextrelation}}
\newcommand{\substrsubform}[1]{\mu_{\zeta,#1}^{\equalsubstr}}

\newcommand{\ie}{i.\,e.\xspace}

\begin{document}
\maketitle
\begin{abstract}
The present paper investigates the dynamic complexity of document spanners, a formal framework for information extraction introduced by Fagin, Kimelfeld, Reiss, and Vansummeren (JACM~2015). We first look at the class of regular spanners and prove that any regular spanner can be maintained in the dynamic complexity class \dynprop. This result follows from work done previously on the dynamic complexity of formal languages by Gelade, Marquardt, and Schwentick (TOCL~2012). 

To investigate core spanners we use \splog, a concatenation logic that exactly captures core spanners. We show that the dynamic complexity class \dyncq is more expressive than \splog and therefore can maintain any core spanner. This result is then extended to show that \dynfo can maintain any generalized core spanner and that \dynfo is more powerful than \splog with negation. 
\end{abstract}
\section{Introduction}\label{sec:intro}
Document spanners where introduced by Fagin, Kimelfeld, Reiss, and Vansummeren~\cite{fag:spa} as a formalization of IBM's information retrieval language AQL. Essentially, they can be explained as a formalism of querying text like one would query a relational database.

The universe of document spanners are \emph{spans}, intervals of positions in a text. For example, if one searches for a word inside a larger text, every match can be understood as being one span inside the text. Spanners generalize this by mapping an input text to a table of spans.

More specifically, the process can be described as follows.
First, \emph{primitive spanners}, so-called \emph{extractors}, are used to convert the input text into tables of spans. 
These extractors can be assumed to be \emph{regex formulas}, which are regular expressions with variables.
The tables can then be combined using relational algebra. As one might expect, different types of spanners allow different choices of operators. 
In this paper, we deal with three types of spanners that were introduced by Fagin et al.~\cite{fag:spa}.
\emph{Regular spanners}, currently the most widely studied in literature, allow the operators $\cup$~(union), $\pi$~(projection), and $\join$ (join). \emph{Core spanners} extend regular spanners by allowing the string equality selection operator~$\select^=$, which allows one to check whether spans describe the same string (but potentially at different places). \emph{Generalized core spanners} then extend these with the set difference~$\diff$.

In the last few years, various aspects of spanners have received considerable attention (see our related work section). The main focus was on evaluation and enumeration of results. But very few papers have considered aspects of maintaining the results of spanners under updates on the input text, and these have only focused on regular spanners.

In this paper, we examine the complexity of this problem from a \emph{dynamic complexity} point of view. 
The classic dynamic complexity setting was independently introduced by 
Dong, Su, and Topor~\cite{don:non} and 
Patnaik and Immerman~\cite{pat:dynfo}.
The ``default setting'' of dynamic complexity assumes a big relational database that is constantly changing (where the updates consist of adding or removing tuples from relations). 
The goal is then to maintain a set of auxiliary relations that can be updated with ``little effort''. As this is a descriptive complexity point of view, little effort is defined as using only first-order formulas. The class of all problems that can be maintained in this way is called \dynfo. 

A more restricted setting is \dynprop, where only quantifier-free formulas can be used. As one might expect, restricting the update formulas leads to various classes between \dynprop and \dynfo. Of particular interest to this paper are \dyncq and \dynucq, where the update formulas are conjuctive queries or unions of conjunctive queries. As shown by Zeume and Schwentick~\cite{zeu:dyncq}, $\dyncq=\dynucq$ holds; but it is open whether these are proper subclasses of \dynfo (see Zeume~\cite{zeu:small} for detailed background information).

As document spanners are defined on words, we adapt the dynamic complexity setting for formal languages by Gelade, Marquardt, and Schwentick~\cite{gel:dyn}. This interprets a word structure as a linear order (of positions in the word) with unary predicates for every terminal symbol. To account for the dynamic complexity setting, positions can be undefined, and the update operations are setting a position to a symbol (an insertion or a symbol change) and resetting a position to undefined (deleting a symbol).

We show that in this setting, regular spanners can be maintained in \dynprop, core spanners in \dynucq (and, hence, by~\cite{zeu:dyncq} in \dyncq), and generalized core spanners in \dynfo. 
Here, the second of these results is the main result of the present paper (the third follows directly from it, and the first almost immediately from~\cite{gel:dyn}).
To achieve it, we do not convert core spanners directly, but use the concatenation logic \splog  as an intermediate model.

\splog  (short for \emph{spanner logic}) was introduced by Freydenberger~\cite{fre:splog} and has the same expressive power as core spanners (under some caveats that we discuss in \cref{sec:splog}). An additional benefit of the main result is that \splog can be used to simplify proofs that languages or word relations can be maintained in \dyncq. 

\subparagraph*{Related work}
Recently, algorithmic and complexity theoretic aspects of evaluation and enumeration of spanners have received a considerable amount of attention, see~\cite{ama:con,  flo:con,  fre:doc,  fre:joi,  fre:splog, los:fou, mat:doc, mor:eng, pet:com,  pet:rec}. But these almost exclusively consider spanners in a static setting.
To the authors' knowledge, the only articles to also examine updates   are Losemann~\cite{los:fou} and  Amarilli, Bourhis, Mengel, and Niewerth~\cite{ama:con}. 
Both do not take  a \dynfo point of view; moreover, both only deal with regular spanners and there is no obvious way to also include the string equalities that are required for core spanners and generalized core spanners.

Doleschal, 
 Kimelfeld, 
 Martens, 
 Nahshon, and
 Neven~\cite{dol:split} 
introduce the notion of split-correctness. Without going into details, this examines spanners for which it is possible to split the input word into subwords on which the spanner is then evaluated. This can be viewed as a special case of update, but again was restricted to regular spanners.

Gelade, Marquardt, and Schwentick~\cite{gel:dyn} examined the dynamic complexity of formal languages. Their result that \dynprop captures the regular languages is the basis for~\cref{prop:regular} in the current paper. 
While they also
established that every context free language is in \dynfo and that every Dyck-language is in \dynqf (\dynprop with auxiliary functions), they did not examine \dynucq and \dyncq, which the present paper does.

Mu{\~{n}}oz, Vortmeier, and Zeume~\cite{mun:dyn} studied the dynamic complexity in a graph database setting, namely for \emph{conjunctive regular path queries (CRQPs)} and \emph{extended conjunctive regular path queries (ECRPQs)}. In particular, Theorem~14 in~\cite{mun:dyn} states that on acyclic graphs, even a generalization of ECRPQs can be maintained in \dynfo. Fagin et al.~\cite{fag:spa}  established that on marked paths (a certain type of graph)  core spanners have the same expressive powers as a CRPQs with string equalities (a fragment of ECRPQs). While marked paths are not acyclic in a strict sense, Section~7 of \cite{fre:splog} proposes a variant of this model that could be directly combined with the construction from~\cite{mun:dyn}. Thus, one could combine these results and observe that core spanners can be maintained in \dynfo. In contrast to this, the present paper allows us to lower the upper bound to \dyncq. Moreover, if one is satisfied with \dynucq, the constructions in the present paper also guarantee that all auxiliary relations only contain active nodes (nodes which carry a letter) of the word-structure, the only exception being the special case where the word-structure represents the empty string. 

\subparagraph*{Structure of the paper} 
Section~\ref{sec:prelim} contains the central definitions.  
Section~\ref{sec:main} establishes dynamic upper bounds for the three central classes of document spanners (regular, core, and generalized core spanners), in particular the main result (Theorem~\ref{splogindyncq}). Section~\ref{sec:rel} further examines the relative expressive powers of core spanners and \dyncq. Section~\ref{sec:conc} concludes the paper.
Some of the longer proofs have been moved to the appendix.
\section{Preliminaries}\label{sec:prelim}
Let $\mathbb{N} := \{ 0 , 1 , 2 \dots \}$ and let $\mathbb{N}_+ := \mathbb{N} \setminus \{0\}$, where $\setminus$ denotes set difference. We write $|S|$ to represent the \emph{cardinality} of a set $S$. We use $\subseteq$ for subset and $\subset$ for proper subset. We denote the powerset of $S$ by $\mathcal{P} (S) $. Let $\emptyset$ be the empty set. If $R$ is a relation of arity $0$, then $R$ is the empty set, or $R$ is the set containing the empty tuple. We define $[n]\df\{ 1 , 2 \dots n \} $. 

 Let $A$ be an alphabet\footnote{We use $A$ here as a generic alphabet since we look at both the alphabet of terminal symbols and the alphabet of variables, and the concepts defined here apply to both.}.
 We write $|w|$ to denote the length of a word $w \in A^*$. The number of occurrences of some $a \in A$ in a word $w \in A^*$ is represented by $| w |_a$. We use $\emptyword$ to denote the empty word. Given two words $u \in A^*$ and $v \in A^*$, we write $u \cdot v$, or simply $uv$ for concatenation. 
If  $w = v_1 u v_2$ where $v_1 \in A^*$ and $v_2 \in A^*$, then $u$ is a \emph{subword} of $w$. We use~$\sqsubseteq$ for subword and $\sqsubset$ for the proper subword relation. If $u$ is not a subword of $w$, we write $u \not\sqsubseteq w$. Let $\Sigma$ be a finite alphabet of so-called \emph{terminal symbols}. Let $\Xi$ be an infinite set of so-called \emph{variables}, which is disjoint from $\Sigma$. Let $\mathcal{L}(A)$ (or $\mathcal{L}(\alpha)$) denote the language of a nondeterministic finite automaton (NFA) $A$ (or of a regular expression $\alpha$).

The rest of this section is structured as follows: First, we define various types of document spanners in \cref{sec:spanners} and equivalent logics (\cref{sec:splog}). After that, we define dynamic complexity, with a particular focus on its application to document spanners (\cref{sec:dyn}).

\subsection{Document Spanners and Spanner Algebra}\label{sec:spanners}

In this section, we introduce document spanners and their representations. We begin with \emph{primitive spanners} (\cref{sec:spanner-rep}) and then combine these to \emph{spanner algebras} (\cref{sec:spannerAlgebra}).
\subsubsection{Primitive Spanner Representations}\label{sec:spanner-rep}
Let $w := a_1 \cdot a_2 \cdots a_n$ be a word, where $n \geq 0$ and $a_1,\dots,a_n \in \Sigma$. A \emph{span} of $w$ is an interval $\spn{i,j}$ with $1 \leq i \leq j \leq n+1$ and $i,j \geq 0$. Given a span $\spn{i,j}$ of a word $w$, we define the subword $w_{\spn{i,j}}$ as $a_i \cdot a_{i+1} \cdots a_{j-1}$.
\begin{example}
Consider the word 
$w\df \mathtt{banana}$.	As $|w|=6$, the spans of $w$ are the $\spn{i,j}$ with $1\leq i \leq j\leq 7$. For example, we have $w_{\spn{1,2}}=\mathtt{b}$ and $w_{\spn{2,4}}=w_{\spn{4,6}}=\mathtt{an}$. Although both spans describe the same subword \texttt{an}, the two occurrences are at different locations (and, thus, at different spans). Analogously, we have $w_{\spn{1,1}}=w_{\spn{2,2}}=\cdots=w_{\spn{7,7}}=\emptyword$, but $\spn{i,i}\neq \spn{i',i'}$ for all distinct $1\leq i,i'\leq 7$.
\end{example}

Let  $V \subseteq \Xi$ and  $w \in \Sigma^*$. A $(V,w)$\emph{-tuple} is a function $\mu$ that maps each  $x\in V$ to a span $\mu(x)$ of $w$. A set of $(V,w)$-tuples is called a  $(V,w)$\emph{-relation}. A \emph{spanner} P is a function that maps every $w \in \Sigma^*$ to a $(V,w)$-relation~$P(w)$. 
We write $\SVars{P}$ to denote the set of variables $V$ of a spanner~$P$. 
Two spanners $P_1$ and $P_2$ are \emph{equivalent} if $\SVars{P_1} = \SVars{P_2}$ and  $P_1(w) = P_2(w)$ holds for all $w \in \Sigma^*$.

In the usual applications of spans and spanners, the word $w$ is some type of text. Hence, we can view a spanner $P$ as mapping an input text $w$ to a $(V,w)$-relation  $P(w)$, which can be understood as a table of spans of $w$. 

To define spanners, we use two types of \emph{primitive spanner representations}, the so-called \emph{regex formulas} and \emph{variable-set automata}. Both extend classical mechanisms for regular languages (regular expressions and NFAs, respectively) with variables.  
\subparagraph*{Regex formulas:}
The syntax of regex formulas is defined by the following
$\alpha \df \emptyset \mid \emptyword \mid a \mid (\alpha \lor \alpha) \mid (\alpha \cdot \alpha) \mid (\alpha)^* \mid x \{ \alpha \}, $
where $a \in \Sigma$ and $x \in \Xi$. We use $\alpha^+$ to denote $\alpha \cdot \alpha^*$. 

Like~\cite{fre:splog}, we define the  semantics of regex formulas using two step-semantics with  \emph{ref-words} (originally introduced by Schmid~\cite{sch:cha} in a different context). A ref-word is a word over the extended alphabet $(\Sigma \cup \Gamma)$ where $\Gamma \df \{ \openvar{x}, \closevar{x} \mid x \in \Xi \}$. The symbols $\openvar{x}$ and $\closevar{x}$ represent the beginning and end of the span for the variable $x$. 
The first step in the definition of semantics is treating each regex formula $\alpha$ as generators of languages of ref-words $\rlang(\alpha)\subseteq (\Sigma \cup \Gamma)^*$, which  is defined by $\rlang(\emptyset) \df \emptyset$, $\rlang(a) \df \{ a \}$ where $a \in \Sigma \cup \{ \emptyword \}$, $\rlang(\alpha_1 \lor \alpha_2) \df \rlang(\alpha_1) \cup \rlang(\alpha_2)$, $\rlang(\alpha_1 \cdot \alpha_2) \df \rlang(\alpha_1) \cdot \rlang(\alpha_2)$, $\rlang(\alpha^*) \df \rlang(\alpha)^*$, and $\rlang(x\{ \alpha \}) \df \openvar{x} \rlang(\alpha) \closevar{x}$.

Let $\SVars{\alpha}$ be the set of all $x \in \Xi$ such that $x\{ \}$ occurs somewhere in $\alpha$. A ref-word $r \in \rlang(\alpha)$ is \emph{valid} if for all $x \in \SVars{\alpha}$, we have that $|r|_{\openvar{x}} = 1$. 
We denote the set of valid ref-words in $\rlang(\alpha)$ as $\validr{\alpha}$ and say that 
a regex formula is \emph{functional} if $\rlang(\alpha) = \validr{\alpha}$. 
We write $\rgx$ for the set of all functional regex formulas. 
By definition,  for  every $\alpha\in\rgx$, every $r \in \validr{\alpha}$, and  every $x \in \SVars{\alpha}$, there is  a unique factorization $r = r_1 \openvar{x} r_2 \closevar{x} r_3$. 

This allows us to define the second step of the semantics, which turns such a  factorization for some variable $x$ into a span $\mu(x)$. To this end, we define a morphism $\clr \colon (\Sigma \cup \Gamma)^* \rightarrow \Sigma^*$ by $\clr(a) \df a $ for $a \in \Sigma$ and $\clr(g) = \emptyword$ for all $g \in \Gamma$. 
For a  factorization $r = r_1 \openvar{x} r_2 \closevar{x} r_3$, $\clr(r_1)$ is the substring of $w$ that appears before $\mu(x)$ and $\clr(r_2)$ is the substring $w_{\mu(x)}$.

We use this for the definition of the semantics as follows: 
For $\alpha \in \rgx$ and $w \in \Sigma^*$, let $V \df \SVars{\alpha}$ and (more importantly) $\validr{\alpha, w} \df \{ r \in \validr{\alpha} \mid \clr(r)=w \}$.

Every $r \in \validr{\alpha,w}$ defines a $(V,w)$-tuple $\mu^r$ in the following way: For every $x \in \SVars{\alpha}$, we use  the unique factorization $r = r_1 \openvar{x} r_2 \closevar{x} r_3$ to define $\mu^r(x) \df \spn{|\clr(r_1)|+1, |\clr(r_1 r_2)|+1	}$. 
The spanner $\spanner{\alpha}$ is then defined by $\spanner{\alpha}(w) \df \{ \mu^r \mid r \in \validr{\alpha,w} \} $ for all $w\in\Sigma^*$.

\subparagraph*{Variable-set automata:} Variable-set automata (short: \emph{vset-automata}) are NFAs that may use variable operations $\openvar{x}$ and $\closevar{x}$ as  transitions. More formally, let  $V \subset \Xi$ be a finite set of variables. 
 A variable-set automaton over $\Sigma$ with variables $V$ is a tuple $A = (Q, q_0, q_f , \delta)$, where $Q$ is the set of states, $q_0 \in Q$ is the initial state, $q_f \in Q$ is the accepting state, and $\delta \colon Q \times (\Sigma \cup \{ \emptyword \} \cup \Gamma_V) \rightarrow \powerset{Q}$ is the transition function with $\Gamma_V \df \{ \openvar{x}, \closevar{x} \mid x \in V \}$. 
 
We define the semantics using a two-step approach analogous to the semantic definition of regex formulas. 
 Firstly, we treat $A$ as an NFA that defines the ref-language defined by $\rlang(A)\df\{ r \in (\Sigma \cup \Gamma_V)^* \mid q_f \in \deltah(q_0, r) \}$, where 
 the function $\deltah \colon Q \times (\Sigma \cup \Gamma_V) \rightarrow \powerset{Q}$ is defined such that for all $p,q \in Q$ and $r \in (\Sigma \cup \Gamma_V)^*$, $q \in \deltah(p,r)$ if and only if there exists a path in $A$ from $p$ to $q$ with the label $r$. 

Secondly, let $\SVars{A}$ be the set of $x \in V$ such that $\openvar{x}$ or $\closevar{x}$ appears in $A$. A ref-word $r \in \rlang(A)$ is \emph{valid} if for every $x \in \SVars{A}$, $|r|_{\openvar{x}} = |r|_{\closevar{x}}=1$, and $\openvar{x}$ always occurs to the left of $\closevar{x}$. Then $\validr{A}$, $\validr{A,w}$ and $\spanner{A}$ are defined analogously to regex formulas. We denote the set of all vset-automata using $\VAset$. As for regex formulas, a vset-automaton $A\in\VAset$ is called \emph{functional} if $\rlang(A)=\validr{A}$.

\begin{example}\label{ex:regexAutomaton}
We define the functional regex formula 
$\alpha\df\Sigma^*\cdot x\{(\mathtt{wine})\lor(\mathtt{cake})\}\cdot\Sigma^*$. We also define the functional vset-automaton~$A$ as follows:
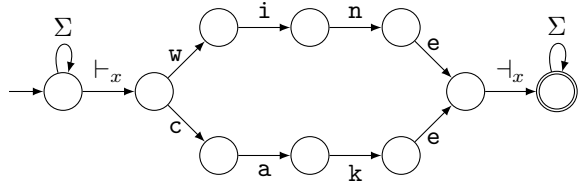
\begin{center}
	\begin{tikzpicture}[on grid, node distance =12mm,every loop/.style={shorten >=0pt}, every state/.style={inner sep=0pt,minimum size=5mm}]
	\node[state,initial text=,initial by arrow] (q0) {};
	\node[state, right= of q0] (q1) {};
	\node[state, above right=of q1] (b1) {};
	\node[state, right=of b1] (b2) {};
	\node[state, right=of b2] (b3) {};

	\node[state, below right=of q1] (c1) {};
		\node[state, right=of c1] (c2) {};
		\node[state, right=of c2] (c3) {};
	\node[state, below right=of b3] (q3) {};
	\node[state,accepting,right=of q3] (q4) {};	\path[->]
	(q0) edge [loop above] node {$\Sigma$} (q0)
	(q0) edge node[above] {$\openvar{x}$} (q1)
	(q1) edge[pos=0.2] node[above] {$\mathtt{w}$} (b1)
	(b1) edge[] node[above] {$\mathtt{i}$} (b2)
	(b2) edge[] node[above] {$\mathtt{n}$} (b3)
	(b3) edge[] node[above] {$\mathtt{e}$} (q3)
	(q1) edge[pos=0.2] node[below] {$\mathtt{c}$} (c1)
(c1) edge[] node[below] {$\mathtt{a}$} (c2)
(c2) edge[] node[below] {$\mathtt{k}$} (c3)
(c3) edge[] node[below] {$\mathtt{e}$} (q3)
	(q3) edge node[above] {$\closevar{x}$} (q4)
(q4) edge [loop above] node {$\Sigma$} (q4)
	;
	\end{tikzpicture}
\end{center}  
For all $w\in\Sigma^*$, we have that $\spanner{\alpha}(w)=\spanner{A}(w)$ contains exactly those $(\{x\},w)$-tuples $\mu$ that have $w_{\mu(x)}=\mathtt{wine}$ or $w_{\mu(x)}=\mathtt{cake}$.
\end{example}

\subsubsection{Spanner Algebra}\label{sec:spannerAlgebra}
We now introduce an algebra on spanners in order to construct more complex spanners.

\begin{definition}
Two spanners $P_1$ and $P_2$ are \emph{compatible} if $\SVars{P_1}=\SVars{P_2}$.
We define  the following algebraic operators for all spanners $P, P_1, P_2$:
\begin{itemize}
	\item If $P_1$ and $P_2$ are compatible, their \emph{union} $(P_1 \cup P_2)$ and their \emph{difference} $(P_1\diff P_2)$ are defined by $(P_1 \cup P_2)(w) \df P_1(w) \cup P_2(w)$ and $(P_1 \diff P_2)(w) \df P_1(w) \diff P_2(w)$.
	\item The \emph{projection} $\quotproj_Y P$ for $Y \subseteq \SVars{P}$ is defined by $\quotproj_Y P(w) \df P|_Y(w)$, where $P|_Y(w)$ is the restriction of all $\mu \in P(w)$ to~$Y$.
	\item The \emph{natural join} $P_1 \join P_2$ is obtained by defining each $(P_1 \join P_2)(w)$ as the set of all  $(V_1 \cup V_2,w)$-tuples $\mu$ for which there exists $\mu_1 \in P_1(w)$ and $\mu_2 \in P_2(w)$ with $\mu|_{V_1}(w) = \mu_1(w)$ and $\mu|_{V_2}(w) = \mu_2(w)$, where $V_i \df \SVars{P_i}$ for $i \in \{ 1,2 \}$.   
	\item For every  $k$-ary relation $R \subseteq (\Sigma^*)^k$ and variables $x_1, \dots, x_k \in \SVars{P}$, the \emph{selection}  $\select^R_{x_1 \dots x_k}P$ is defined by
	$\select^R_{x_1 \dots x_k} P(w)\df \{\mu \in P(w) \mid (w_{\mu(x_1)} , \dots , w_{\mu(x_k)}) \in R\}$ for $w\in\Sigma^*$.
\end{itemize}
Let  $\SVars{P_1 \cup P_2} \df \SVars{P_1\diff P_2}\df\SVars{P_1}=\SVars{P_2}$,   $\SVars{\quotproj_Y P} \df Y$,  $\SVars{P_1 \join P_2} \df \SVars{P_1} \cup \SVars{P_2}$, and $\SVars{\select^R_{x_1 \dots x_k}} \df \SVars{P}$.
\end{definition}

 Note that the relations $R$ in the selection are usually infinite; and they are never considered part of the input.

Let $O$ be a spanner algebra and let $C$ be a class of primitive spanner representations, then we use $C^O$ to denote the set of all spanner representations that can be constructed by repeated combinations of the symbols for the operators from $O$ with the spanner representation from~$C$. We denote the closure of $\spanner{C}$ under the spanner operators $O$ as $\spanner{C^O}$. 
\begin{example}

Let $\alpha_1\df \Sigma^*x\{\Sigma^*\}\Sigma^*y\{\Sigma^*\}\Sigma^*$ and $\alpha_2\df\Sigma^*\cdot x\{(\mathtt{wine})\lor(\mathtt{cake})\}\cdot\Sigma^*$ (recall Example~\ref{ex:regexAutomaton}). We combine the two regex formulas into a core spanner $P \df \pi_{x}\select^=_{x,y} (\alpha_1\join\alpha_2)$. Then $\spanner{P}(w)$ contains all $(\{x\},w)$-tuples $\mu$ such that $w_{\mu(x)}$ is an occurrence of \texttt{wine} or \texttt{cake} in $w$ that is followed by another occurrence of the same word.
\end{example}

Like  Fagin et al. \cite{fag:spa}, we are mostly concerned with string equality selections $\select^=$. Following~$\cite{fag:spa,pet:rec}$, we focus on the class  of \emph{regular spanners} $\spanner{\RGXreg}$, the class of  \emph{core spanners}\footnote{As this class captures the core functionality of SystemT.}  $\spanner{\RGXcore}$ 
and the class of \emph{generalized core spanners}  $\spanner{\RGXcored}$, where 
$\spanreg\df\{\quotproj,\cup,\join\}$ and  $\core \df \{ \quotproj, \select^=, \cup, \join\}$. 
As shown in~\cite{fag:spa}, we have
\[ \spanner{\RGXreg}=\spanner{\VAsetreg} =\spanner{\VAset}  \subset  \spanner{\RGXcore}=\spanner{\VAsetcore}\subset \spanner{\RGXcored}=\spanner{\VAsetcored}. \]
In other words, there is a proper hierarchy of regular, core, and generalized core spanners; and for each of the classes, we can choose regex formulas or vset-automata as primitive spanner representations. As shown in~\cite{fre:splog}, functional vset-automata have the same expressive power as vset-automata in general. The size difference can be exponential, but this does not matter for the purpose of the present paper.

\subsection{Spanner Logic}\label{sec:splog}

In this section, we define \splog (spanner logic) and relate it to spanners.
\splog is a fragment of \ECrtext, the existential theory of concatenation with regular constraints (a logic that is built around the concatenation operator).
It was introduced by Freydenberger~\cite{fre:splog} and has the same expressive power as core spanners; and conversions between both models are possible in polynomial time. 
To define \splog, we first introduce \emph{word equations}.

A \emph{pattern} $\alpha$ is a word from $(\Sigma \cup \Xi)^*$. In other words, patterns may contain variables and terminal symbol. A \emph{word equation} is a pair of patterns $(\eta_L, \eta_R)$, which are called the \emph{left} and \emph{right} side of the equation, respectively. We usually write a word equation as $\eta_L \weqeq\eta_R$. The set of all variables in a pattern $\alpha$ is denoted by $\var(\alpha)$. This is extended to word equations $\eta = (\eta_L , \eta_R)$ by $\var(\eta) := \var(\eta_L) \cup \var(\eta_R)$.

A \emph{pattern substitution} is a morphism $\sigma : (\Sigma \cup \Xi)^* \rightarrow \Sigma^*$ such that $\sigma(a) = a$ holds for all $a \in \Sigma$. As every substitution $\sigma$ is a morphism, we have $\sigma(\alpha_1\cdot\alpha_2)=\sigma(\alpha_1)\cdot\sigma(\alpha_2)$ for all patterns $\alpha_1$ and $\alpha_2$. Hence, to define $\sigma$, it suffices to define $\sigma(x) \in \Sigma^*$ for all $x \in \Xi$.

The main idea of \splog is choosing a special main variable $\mv$ that shall correspond to the input string of a spanner. \splog is then an existential-positive logic over words, where the atoms are regular predicates or word equations of the form $\mv\weqeq \eta_R$. Formally, 
we define  syntax and semantics as follows:
\begin{definition}
\label{def:splog}
	Let $\mv \in \Xi$. Then $\splog(\mv)$, the set of all \splog-formulas with main variable~$\mv$, is defined recursively as containing the following formulas:
	\begin{description}
		\item[B1.] $(\mv \weqeq \eta_R)$ for every $\eta_R \in  (\Xi\cup\Sigma)^*$.
		\item[R1.] $(\varphi_1 \land \varphi_2)$ for all $\varphi_1, \varphi_2 \in \splog(\mv)$.
		\item[R2.] $(\varphi_1 \lor \varphi_2)$ for all  $\varphi_1, \varphi_2 \in \splog(\mv)$ with $\fvar(\varphi_1) = \fvar(\varphi_2)$.
		\item[R3.] $ \exists x \colon \varphi$ for all $\varphi \in \splog(\mv)$ and $x \in \fvar(\varphi) \setminus \{ \mv \}$.
		\item[R4.] $(\varphi \land \constr_A(x))$ for every $\varphi \in \splog(\mv)$, every $x \in \fvar(\varphi)$, and every NFA~$A$.
	\end{description}

Let $\fvar(\varphi)$ be $\fvar(\eta) := \var(\eta)$, $\fvar(\varphi_1 \land \varphi_2) := \fvar(\varphi_1 \lor \varphi_2) := \fvar(\varphi_1) \cup \fvar(\varphi_2)$, $\fvar(\exists x \colon \varphi) := \fvar(\varphi) \setminus \{ x \}$, and $\fvar(\varphi \land \constr_A(x)) := \fvar(\varphi)$.

For every  pattern substitution~$\sigma$ and every $\varphi\in\splog(\mv)$, we define $\sigma\models\varphi$  as follows:
\begin{itemize}
	\item $\sigma\models (\mv \weqeq \eta_R)$ if  $\sigma(\mv)=\sigma(\eta_R)$,
	\item $\sigma\models (\varphi_1 \land \varphi_2)$  if $\sigma \models \varphi_1$ and $\sigma \models \varphi_2$; and $\sigma \models (\varphi_1 \lor \varphi_2)$ is defined analogously,
	\item $\sigma \models \exists x \colon \varphi$ if  $\subst{\sigma}{x}{w} \models \varphi$ for some $w \in \Sigma^*$, where  $\subst{\sigma}{x}{w}(x) \df w$ and $\subst{\sigma}{x}{w}(y) = \sigma(y)$ if~$y \neq x$,
	\item $\sigma \models (\varphi \land \constr_A(x))$ if $\sigma \models \varphi$ and $\sigma(x) \in \lang(A)$.
\end{itemize}
\end{definition}

Let $\splog$ be the union of all $\splog(\mv)$ with $\mv\in\Xi$. 
We  add and omit parentheses,  as long as the meaning remains unambiguous. 
We also allow constraints of the form $\constr_{\alpha}(x)$, where $\alpha$ is a regular expression.
For readability, we use $\varphi(\mv; x_1, x_2 \dots  x_k)$ to express that the \splog-formula $\varphi$ has  the main variable $\mv$ and free variables $\{x_1 , x_2 \dots x_k \}$.
As a convention, assume that no word equation $(\mv \weqeq \eta_R)$ has the main variable $\mv$ occur in the right side; that is, that $|\eta_R|_{\mv}=0$ holds. 
\begin{example}
For the \splog-formula $\varphi(\mv)\df \exists x\colon \bigl((\mv \weqeq xxx) \land \constr_{\mathtt{ab}^*}(x) \bigr)$, we have  $\sigma\models\varphi$ if and only if $\sigma(\mv)=www$ for some $w\in \mathtt{ab}^*$. 
\end{example}

We also extend the definition of \splog to \splogneg, which we call \splog \emph{with negation}. 
\begin{definition}
Let $\mv \in \Xi$. Then $\splogneg(\mv)$, the set of \splogneg-formulas with the main variable $\mv$, is defined by extending Definition \ref{def:splog} with the additional rule that if $\varphi \in \splogneg(\mv)$, then $(\neg\varphi) \in \splogneg(\mv)$, with $\fvar(\varphi) = \fvar(\neg\varphi)$. We define $\sigma\models\neg\varphi$ as:
\begin{itemize}
	\item $\sigma(x) \sqsubseteq \sigma(\mv)$ for all $x \in \fvar(\varphi)$, and 
	\item $\sigma \models \varphi$ does not hold.
\end{itemize}
\end{definition}

To compare the expressive power of \splog and document spanners, we need to overcome the difficulty that the former reasons about words, while the latter reason over positions in an input word. To this end, we use the following notion that was introduced by Freydenberger and Holldack~\cite{fre:doc} in the context of \ECrtext. 
\begin{definition}
Let $\varphi \in \splog$ with $\fvar(\varphi) \df \{ W \} \cup \{ x_p, x_c \mid x \in \SVars{P}  \}$. Let $P$ be a spanner. Let $\spanner{\varphi}(w)$ denote the set of all $\sigma$ such that $\sigma \models \varphi$ and $\sigma(\mv) = w$. We then say that $\varphi$ \emph{realizes} $P$ if for all $w \in \Sigma^*$, we have 
$\sigma \in \spanner{\varphi}(w)$ if and only if $\mu \in P(w)$ where for each $x \in \SVars{P}$ and $\spn{i,j} \df \mu(x)$, both $\sigma(x_p) = w_{\spn{1,i}}$ and $\sigma(x_c) =  w_{\spn{i,j}}$.
\end{definition}

Intuitively, this definition uses two main ideas: Firstly, the spanner's input word $w$ is represented by the main variable $\mv$. Secondly, every spanner variable $x$ is represented by two \splog-variables $x_p$ and $x_c$, such that in each $(V,w)$-tuple $\mu$, we have that $x_c$ contains the actual content $w_{\mu(x)}$ and $x_p$ contains the prefix of $w$ before the start of $\mu(x)$.

As shown in Section~4.1 of~\cite{fre:splog}, under this lens, \splog has exactly the same expressive power as $\spanner{\RGXcore}$ (the core spanners), and \splogneg  exactly the same as $\spanner{\RGXcored}$ (the~generalized core spanners).

One of the central questions in~\cite{fag:spa,fre:splog} is which relations $R$ can be added to spanners or \splog without increasing the expressive power (using $\select^R$ or a new constraint symbol for $R$, respectively). This is reflected in the notion of \emph{selectable relations}. 
A relation $R\subseteq (\Sigma^*)^k$ is called \emph{\splog-selectable}
 if for every $\varphi\in\splog(\mv)$ and every sequence $\vec{x}=(x_1,\ldots,x_k)$ of variables with $x_1,\ldots,x_k\in\fvar(\varphi)\diff\{\mv\}$, there is a \splog-formula $\varphi^{R}_{\vec{x}}$ with $\fvar(\varphi)=\fvar(\varphi^{R}_{\vec{x}})$, and  $\sigma\models \varphi^{R}_{\vec{x}}$ if and only if $\sigma\models\varphi$ and $(\sigma(x_1),\ldots,\sigma(x_k))\in R$.
This is equivalent to the analogously defined notion of core spanner selectable relations, see Section~5.1 of~\cite{fre:splog} for details. 
We shall use selectability both in the way to our main result (namely, in Lemma~\ref{lem:splogDynCQ}) and for further observations in \cref{sec:rel}.

\subsection{Dynamic complexity}\label{sec:dyn}

Our definitions of dynamic complexity are based on the setting of dynamic formal languages as described by Gelade, Marquardt, and Schwentick \cite {gel:dyn}. In this setting, strings are modeled by a relational structure. Insertions and deletions of symbols can be performed on this structure and (auxiliary) relations are \emph{maintained} by logic formulas, called \emph{update formulas}. We extend this with a predetermined relation which is maintained to hold the result of some spanner performed on the current word. The idea of dynamic complexity, which was introduced by Patnaik and Immerman \cite{pat:dynfo}, is to have dynamic descriptive complexity classes based upon the logic needed to maintain a relation, or in our case a spanner. We now formally define these concepts.

Let $\Sigma$ be a fixed and finite alphabet of terminal symbols. We represent words using a \emph{word-structure}. A word-structure has a fixed and finite set known as the domain $\worddomain := [n+1]$ as well as a 2-ary order relation $<$ on $\worddomain$. 
We use the shorthands $x \leq y$ for $(x<y) \lor (x \logeq y)$. We have in our word-structure the constant $\$$ which is interpreted by the element $n+1$, the $<$-maximal element of $\worddomain$. 
This <-maximal element marks the end of the word structure and is required for dynamic spanners, which are defined later. For each symbol $\zeta \in \Sigma$ the word-structure has a unary relation $R_{\zeta}(i)$ and there is at \emph{most} one $\zeta \in \Sigma$ such that $R_{\zeta}(i)$ for $i \in [n]$. 
If we have $R_{\zeta}(i)$ then we write $w(i) = \zeta$, otherwise we write $w(i) = \emptyword$. If $w(i) \neq \emptyword$ for some $i \in \worddomain$, then we call $i$ a \emph{symbol-element}.

Given a word-structure $\wordstruc$, the word that $\wordstruc$ represents is denoted by $\dynword(\wordstruc)$ and this is defined as $\dynword(\wordstruc) \df w(1) \cdot w(2) \cdots w(n)$. Since for some $j \in \worddomain$ it could be that $w(j) = \emptyword$, it follows that the length of the word $\dynword(\wordstruc)$ is likely to be less than $n$. Let $w \df \dynword(\wordstruc)$, we write $w[i,j]$ to represent the subword $w[i,j] \df w(i) \cdot w(i+1) \cdots w(j)$ where $i,j \in \worddomain$ such that $i<j$. 

We now define the set of \emph{abstract updates} $\Delta := \{ \absins{\zeta} \mid \zeta \in \Sigma \} \cup \{ \absreset \}$. A \emph{concrete update} is $\ins{\zeta}{i}$ or $\reset{i}$, for some $i \in \worddomain \setminus \{ \$ \}$ and $\zeta \in \Sigma$. The difference between abstract updates and concrete updates is that concrete updates can be \emph{performed} on a word-structure. Given a word-structure with a domain of size $n$, we use $\Delta_n$ to represent the set of possible concrete updates. For some $\unknownupdate  \in \Delta_n$, we denote the word-structure $\wordstruc$ after an update is performed by $\unknownupdate (\wordstruc)$ and this is defined as:
\begin{itemize}
	\item If $\unknownupdate  = \ins{\zeta}{i}$, then $R_{\zeta}(i)$ is true and $R_{\zeta'}(i)$ is false for all $\zeta' \in \Sigma$ where $\zeta \neq \zeta'$.
	\item If $\unknownupdate  = \reset{i}$ then $R_{\zeta}(i)$ is false for all $\zeta \in \Sigma$.
\end{itemize}

All other elements keep the symbol they had before the update. For $k \geq 1$, let $\unknownupdate ^* := \unknownupdate _1, \unknownupdate _2, \dots \unknownupdate _k$ be a sequence of updates. We use $\unknownupdate ^*(\wordstruc)$ as a short hand to represent $\unknownupdate _k ( \dots ( \unknownupdate_2  (\unknownupdate _1 (\wordstruc) ) ) \dots )$. We place the restriction that updates must \emph{change} the string. We do not allow $\reset{i}$ if $w(i) = \emptyword$ and we do not allow $\ins{\zeta}{i}$ if $w(i) = \zeta$.

\begin{example}
Given a word-structure $\wordstruc$ over the alphabet $\Sigma \df \{ a,b \}$ with domain $\worddomain = [6]$, where $6 = \$$. If we have that $R_a = \{ 2,4 \}$ and $R_b \df \{ 5 \}$, it follows that $\dynword(\wordstruc) = aab$. Performing the operation $\ins{b}{1}$ would give us an updated word of $baab$. Say if we then perform $\reset{4}$ on our new word structure, we would have the word $bab$.
\end{example}

We define the \emph{auxiliary structure} $\auxstruc$ as a set of relations over the domain of $\wordstruc$. A \emph{program state} $\programstate := (\wordstruc, \auxstruc)$ is a word-structure and an auxiliary structure. An \emph{update program} $\updateprogram$ is a finite set of update formulas, which are of the form $\updateformula{R}{\mathsf{op}}{y ; x_1, \dots , x_k }$. We have an update formula for each $R \in \auxstruc$ and $\mathsf{op} \in \Delta$. An update, $\mathsf{op}(i)$, performed on $\programstate$ yields $\programstate' = (\unknownupdate (\wordstruc), \auxstruc') $ where all relations $R' \in \auxstruc'$ are defined by $R' := \{ \vec{j} \mid \bar{\programstate} \models \updateformula{R}{op}{i;\vec{j}} \}$, where $\vec{j}$ is a $k$-tuple (where $k$ is the arity of $R$) and where $\bar{\programstate} \df (\unknownupdate(\wordstruc), \auxstruc)$. 

We use $w$ to denote $\dynword(\wordstruc)$ for some word structure $\wordstruc$ and we use $w'$ for $\dynword(\unknownupdate(\wordstruc))$ where $\unknownupdate \in \Delta_n$ is some update performed on $\wordstruc$.

Given some $x \in \worddomain$ where $w(x) \neq \emptyword$, we write that $\position{w}{x}= 1$ if for all $x' \in\worddomain$ where $x' < x$ we have that $w(x') = \emptyword$. Let $z,y$ be elements from the domain such that $z<y$ and $w(z) \neq \emptyword$ and $w(y) \neq \emptyword$. If for all $x \in \worddomain$ where $z<x<y$ we have that $w(x) = \emptyword$ then $ \position{w}{y}= \position{w}{z}+ 1$. We write $x \nextsym y$ if and only if $\position{w}{y}= \position{w}{x}+1$. If it is not the case that $x \nextsym y$ then we write $x \not\nextsym y$. 

For every spanner $P$ with $\SVars{P} \df \{ x_1, x_2 \dots x_k \}$ and every word-structure $\wordstruc$, the spanner relation $R^P$ is a $2k$-ary relation over $\worddomain$ where each spanner variable $x_i$ is represented by two components $\openspanvar{x_i}$ and $\closespanvar{x_i}$. 
We obtain $R^P$ on $\wordstruc$ by converting each $\mu\in P(w)$ into a $2k$-tuple $(\openspanvar{x_1}, \closespanvar{x_1}, \openspanvar{x_2}, \closespanvar{x_2} \dots \openspanvar{x_k}, \closespanvar{x_k})$, where for each $i\in[k]$, we have $\mu(x_i)= \spn{\position{w}{\openspanvar{x_i}},\position{w}{\closespanvar{x_i}}}$. The only exception is if $\mu(x_i) = \spn{j,k}$ and $k > |w|$ then $\closespanvar{x_i} = \$$ for such a tuple $(\openspanvar{x_1}, \closespanvar{x_1}, \openspanvar{x_2}, \closespanvar{x_2} \dots \openspanvar{x_k}, \closespanvar{x_k})$. In Example \ref{spanrel:example} we give a spanner represented by a regex formula and show the corresponding spanner-relation on a word-structure.

\begin{definition}
A dynamic program is a triple, containing:
\begin{itemize}
	\item $\updateprogram$ - an update program over $(\wordstruc, \auxstruc)$.
	\item $\mathsf{INIT}$ - a first-order initialization program.
	\item $R^P \in \auxstruc$ - a designated spanner-relation.
\end{itemize}
\end{definition}

For each $R \in \auxstruc$, we have some $\psi_R(\vec{j}) \in \mathsf{INIT}$ which defines the initial tuples of $R$ (before any updates to the input structure occur). Note that $\vec{j}$ is a $k$-tuple where the arity of $R$ is $k$. For our work $\psi_R$ is a first-order logic formula. 

A dynamic program \emph{maintains} a spanner $P$ if we have that $R^P \in \auxstruc$ always corresponds to $P(\unknownupdate^* (\wordstruc)) $. We can then extend this to saying that we maintain a \emph{relation} if there is a designated $R \in \auxstruc$ which is always equivalent to some relation where the relation is defined in terms of the input word.

\begin{example}\label{spanrel:example}
Consider the regex formula $ \alpha \df \Sigma^* \cdot x\{ a \cdot b \} \cdot \Sigma^*$ where $a,b \in \Sigma$ and $x \in \Xi$. Now consider the following word-structure:
\begin{center}
	\begin{tabular}{ccccccc}
		1&2&3&4&5&6&$\$$\\
		$a$&$\emptyword$ & $b$ & $\emptyword$ & $a$ & $\emptyword$ & $\emptyword$
	\end{tabular}
\end{center}
Note that the top row is the elements of the domain in order, and the bottom row is the corresponding symbols. If we maintain the \emph{spanner relation} of $\alpha$, given the word-structure above, we have the relation $R^P \in \auxstruc$ such that $R^P \df \{ (1,5) \}$. Now assume we perform the update $\ins{b}{6}$. The word-structure is now in the following state:

\begin{center}
	\begin{tabular}{ccccccc}
		1&2&3&4&5&6&$\$$\\
		$a$&$\emptyword$ & $b$ & $\emptyword$ & $a$ & $b$ & $\emptyword$
	\end{tabular}
\end{center}

It must be that $\updateformula{R^P}{\absins{b}}{6;x,y}$ updates the relation $R^P$ to $\{ (1,5), (5,\$) \}$ for us to correctly maintain the spanner.

\end{example}

\begin{definition}
\dynfo is the class of all relations which can be maintained by update formulas which are defined using first-order logic. \dynprop is a subclass of \dynfo where all the update formulas are quantifier-free. 
\end{definition}

A first-order formula is a \emph{conjunctive query}, or CQ for short, if it is built up from atomic formulae, conjunction and existential quantification. We also have unions of conjunctive queries, or UCQ for short, which allows for the finite disjunction of conjunctive queries. We~therefore have the classes \dyncq and \dynucq which use conjunctive queries and unions of conjunctive queries as update formulas respectively.	

For this work, we assume that the input structure is initially empty and that every auxiliary relation is initialized by some first-order initialization. This is to allow us to use the result from Zeume and Schwentick~\cite{zeu:dyncq} that $\dynucq = \dyncq$. However, in our work we only require a very weak form of initialization and hence if $\dynucq$ is sufficient, one could define the precise class needed for the precomputation. We do not do this as the dynamic complexity class needed to \emph{maintain} a spanner is the main focus of this work\footnote{As helpfully pointed out by one of the anonymous reviewers of this paper.}. 

For the proofs in the present paper, one could change the setting by allowing the insertion of unmarked nodes at any point of the word-structure (with an update to the <-relation), given that the word is non-empty. The auxiliary relations in our proofs do not operate on unmarked nodes and do not need to be updated after this. In the same way, we can remove unmarked nodes. However, the present paper does not look at this setting.

\section{Core Spanners are in \dyncq}\label{sec:main}
\newcommand{\limitforreg}{\substack{
p,q \in Q, \; f \in F \\
\delta(s,\xi_1)=p \\ 
\delta(q,\xi_2)=f}}

In this section, we first look at the dynamic complexity of regular spanners. We show that any regular spanner can be maintained by a \dynprop program. We then turn our attention to the main result of this paper, that any core spanner can be maintained by a \dyncq program. In doing so, we also show that \dyncq is at least as expressive as \splog. We then extend this result to show that \dynfo is at least as powerful as \splog with negation, and therefore any generalized core spanner can be maintained in \dynfo.

\begin{proposition}\label{prop:regular}
Regular spanners can be maintained in \dynprop.
\end{proposition}
\newcommand{\twolinelimit}{\substack{
p \in Q, \\
\delta(s_i,\zeta)=p 
}}

\begin{proof}
Due to the work done by Fagin et al. \cite{fag:spa} we can assume that our vset-automaton is a so called \emph{vset-path union}. We define a vset-path as an ordered sequence of regular deterministic finite automata $A_1$, $A_2$, \dots $A_n$ for some $n \in \mathbb{N}$. Each automaton $A_i$ is of the form $(Q, q_o, F, \delta)$ where $Q$ is the set of states, $q_0 \in Q$ is the initial state, $F$ is the set of accepting states, and $\delta$ is the transition function of the form $\delta \colon Q \times \Sigma \rightarrow Q$. We have the extra assumption that each $f \in F$ only has incoming transitions. All automata, $A_1, A_2, \dots A_n$ share the same set of input symbols $\Sigma$.

Let $A$ be a vset-path. In $A$, each automata $A_i$ where $1 < i \leq n$, the initial state for $A_i$ has incoming transitions from each accepting state from the automaton $A_{i-1}$. These extra transitions between the sequence of automata are labeled, $\openvar{x}$ or $\closevar{x}$ where $x \in \SVars{A}$. We treat the vset-path as a regular vset-automaton and all semantics follow from the definitions in Section \ref{sec:spanner-rep}. We can assume that $A$ is functional \cite{fre:splog}.

Any vset-automaton can be represented as a union of vset-paths~ \cite{fag:spa}. Therefore to prove that any regular spanner can be maintained in \dynprop, it is sufficient to prove that we can maintain a spanner represented by a vset-path, since union can be simulated via disjunction. 

Let $A$ be a vset-path. From Gelade et al. \cite{gel:dyn}, we know that the following relations can be maintained in \dynprop:
\begin{itemize}
\item For any pair of states $p,q \in Q$, $ R_{p,q} \df \{ (i,j) \mid i<j \text{ and } \delta^*(p,w[i+1,j-1]) = q \}$.
\item For each state $q$, $R^I_q \df \{ i \mid \delta^*(q_0,w[1,j-1])=q \}$.
\item For each state $p$, $R^F_p \df \{ j \mid \delta^*(p,[i+1,n]) \in F \}$.
\end{itemize}

We maintain these relations for the vset-path. Some work is needed to deal with the transitions labeled $\openvar{x}$ and $\closevar{x}$. Let $A_i$ and $A_{i+1}$ be two sub-automata such that $1 \leq i < n$, where $n$ is the number of sub-automata. Let $s_{i}$ and $s_{i+1}$ be the starting states for automata $A_i$ and $A_{i+1}$ respectively. Likewise, let $F_i$ and $F_{i+1}$ be the sets of accepting states of $A_i$ and $A_{i+1}$ respectively. The intuition is that if $R_{p, f_i}(x,y)$ where $f_i \in F_i$ holds, then so should $R_{p,s_{i+1}}(x,y)$ since the transition from an accepting state of $A_i$ to the starting state of $A_{i+1}$ is $\openvar{x}$ or $\closevar{x}$. To achieve this, we have the following update formula for $R_{p,s_{i+1}}$
\[
\updateformula{R_{p,s_{i+1}}}{\unknownupdate}{u;x,y} \df \bigvee\limits_{f \in F_{i}} \updateformula{R_{p,f}}{\unknownupdate}{u;x,y}.
\]

We do the analogous for $R^I_q$ and $R^F_p$. If $R^I_{f_{i}}(x)$ holds for any $f_i \in F_i$, then so should $R^I_{s_{i+1}}(x)$. Similarly, if $R^F_{s_{i+1}}(x)$ holds, then so should $R^F_{f_{i}}(x)$ for all $f_i \in F_i$. To achieve this, we proceed analogously to what was done for $\updateformula{R_{p,s_{i+1}}}{\unknownupdate}{u;x,y}$. We also maintain the 0-ary relation $\mathsf{ACC}$ to say whether the word-structure is a member of the language of the vset-path.

We will now give two useful subformulas
\[ \psi^{k'} \df \bigwedge\limits_{1 \leq i \leq k'} \Big( \bigvee\limits_{\zeta \in \Sigma} \big( R_\zeta(\openspanvar{x_i}) \land {R'}^I_{s_i}(x_i^o) \land \bigvee\limits_{\twolinelimit} \big( R'_{p,s_{i+1}}(\openspanvar{x_i},\closespanvar{x_i}) \land \bigvee\limits_{\zeta_2 \in \Sigma}R_{\zeta_2}(\closespanvar{x_i} ) \big) \big) \Big) \]

and 
\[ \psi^{\$} \df \bigvee\limits_{\zeta \in \Sigma} \big( R_\zeta(\openspanvar{x_k}) \land {R'}^I_{s_k}(x_k^o) \land {R'}^I_{F_k}(x_k^c) \land (x_k^c \logeq \$ ) \big). \]

We now give the update formula to maintain a vset-path spanner $A$ with variables $\SVars{A} \df \{x_1, x_2, \dots , x_k \}$
\[ \updateformula{R^A}{\unknownupdate}{u; \openspanvar{x_1},\closespanvar{x_1}, \dots , \openspanvar{x_k}, \closespanvar{x_k} } \df \updateformula{\mathsf{ACC}}{\unknownupdate}{u} \land \big( \psi^{k} \lor (\psi^{k-1} \land \psi^{\$}) \big). \]

Note that, without loss of generality, $R'_{p,q}(x,y)$ is used as a shorthand for $\updateformula{R_{p,q}}{\unknownupdate}{u;x,y}$.
\end{proof}

Since Gelade et al. \cite{gel:dyn} proved that \dynprop maintains exactly the regular languages, it is somewhat unsurprising that we can extend that result to regular spanners. Some work is needed in order to maintain the relation of the spanner, which is why a formal proof of~\cref{prop:regular} is given.

\begin{definition}
The next symbol relation is defined as $\nextrelation := \{ (x,y) \in \worddomain^2 \mid x \nextsym y \}$.
\end{definition}

As stated in Section \ref{sec:dyn}, it is known that $\dyncq=\dynucq$ and therefore to show that a relation can be maintained in \dyncq, it is sufficient to show that the relation can be maintained with \ucq update formulas. We use this to prove many of our results.

\begin{lemma}
\label{lemma:next}
The next symbol relation can be maintained in \dyncq.
\end{lemma}

To prove Lemma \ref{lemma:next}, we maintain the relations $\firstrel \df \{ x \in \worddomain \mid \position{w}{x}=1 \}$ and $\lastrel \df \{ x \in \worddomain \mid \position{w}{y} = |w| \}$. Note that these relations would be undefined for an empty input structure (because $\position{w}{x}$ is undefined). Hence we have that if $|w|=0$ then $x \in \firstrel$ if and only if $x =\$$, and $y \in \lastrel$ if and only if $y$ is the $<$-minimal element. This requires the initialization of $\firstrel \df \{ \$ \}$ and $\lastrel \df \{ 1 \}$. This is the only initialization required in our work, however the stated \emph{first-order initialization} of auxiliary relations is needed to ensure $\dynucq = \dyncq$. 

\begin{example}
Consider the following word-structure: 
\begin{center}
	\begin{tabular}{ccccccc}
		1&2&3&4&5&6&$\$$\\
		$\emptyword$ & $a$ & $b$ & $\emptyword$ & $b$ & $\emptyword$ & $\emptyword$
	\end{tabular}
\end{center}
We have that $\firstrel = \{ 2 \}$ and $\lastrel = \{ 5 \}$ and $\nextrelation = \{ (2,3), (3,5)\}$.
\end{example}

We will now give an idea for the proof of~\cref{lemma:next}. Let $u$ be the node which is being updated. For insertion, if $x \nextsym y$ and $x < u < y$ then $x \newnextsym u \newnextsym y$. If $\firstrel(x)$ and $u < x$, then $\firstrel'(u)$ and $u \newnextsym x$. The analogous is done if $\lastrel(x)$ and $u> x$. For deletion, if $x \nextsym u \nextsym y$ then $x \newnextsym y$. The full proof also looks at when $x \nextsym y$ and $x \newnextsym y$ (for example when $u < x$ or when $u > y$). See the appendix for the proof.

\begin{definition}
The equal substring relation, $\equalsubstr$, is the set of 4-tuples $(\oneopen,\oneclose,\twoopen,\twoclose)$ such that $w[\oneopen,\oneclose]=w[\twoopen,\twoclose]$, 
$\oneclose < \twoopen$, 
and $w[z]\neq\emptyword$ for all $z\in\{\oneopen,\oneclose,\twoopen,\twoclose\}$. 
\end{definition}

Less formally, we have that if $(\oneopen,\oneclose,\twoopen,\twoclose) \in \equalsubstr$ then the word $w[\oneopen,\oneclose]$ is equal to the word $w[\twoopen,\twoclose]$. For our uses, we do not want these subwords to overlap, hence the constraint  $\oneclose < \twoopen$. The reason for this will become clear later on when we look at maintaining pattern languages. We also wish that each tuple represents a unique pair of subwords, therefore we have that $\oneopen$, $\oneclose$, $\twoopen$, and $\twoclose$ each have symbols associated to them.

\begin{example}
Consider the following word-structure:
\begin{center}
	\begin{tabular}{ccccccccccc}
		1&2&3&4&5&6&7&8&9&10& $\$$ \\
		$a$ & $\emptyword$ & $\emptyword$ & $b$ & $a$ & $\emptyword$ & $b$ & $\emptyword$ & $a$ & $b$ & $\emptyword$
	\end{tabular}
\end{center}

The equal substring relation for this structure is
$
\equalsubstr = \{  (1,1,5,5), (1,1,9,9), (4,4,7,7),$ $(4,4,10,10), 
(5,5,9,9), (7,7,10,10), (1,4,5,7), (1,4,9,10), (4,5,7,9), (5,7,9,10) \}.
$

Although $w[3,5] = w[7,9]$, this does not imply $(3,5,7,9) \in \equalsubstr$ because $w[3] = \emptyword$. We also do not have $(9,10,5,7) \in \equalsubstr$ because $10 > 5$.
\end{example}

\begin{lemma}\label{lemma:eqsubstr}
The equal substring relation can be maintained in \dyncq.
\end{lemma}
We now give a proof idea for~\cref{lemma:eqsubstr}. There are four main cases for the tuple $(x_1,y_1,x_2,y_2)$ we examine in the full proof.
\begin{itemize}
\item Case 1: $w[x_1,y_1] = w[x_2,y_2]$ and $w'[x_1,y_1] \neq w'[x_2,y_2]$.
\item Case 2: $w[x_1,y_1] = w[x_2,y_2]$ and $w'[x_1,y_1] = w'[x_2,y_2]$.
\item Case 3: $w[x_1,y_1] \neq w[x_2,y_2]$ and $w'[x_1,y_1] = w'[x_2,y_2]$.
\item Case 4: $w[x_1,y_1] \neq w[x_2,y_2]$ and $w'[x_1,y_1] \neq w'[x_2,y_2]$.
\end{itemize}

Where we assume that $y_1 < x_2$. One can see that the main case out of these four is Case~3. One of the interesting sub-cases of Case 3 is illustrated in~\cref{fig:word}. Here, one can think of the new symbol at node $u$ as a ``bridge'' between the two equal substrings $w[x_1,v_1]$ and $w[x_2,v_3]$ (which are the word $w_1$) and the equal substrings $w[v_2,y_1]$ and $w[v_4,y_2]$ (which are the word $w_2$). Hence, after the update we have that $w'[x_1,y_1]=w'[x_2,y_2]$ even though $w[x_1,y_1] \neq w[x_2,y_2]$ (under the assumptions that $w(v) = a$, $v_1 \newnextsym u \newnextsym v_2$ and that $v_3 \newnextsym v \newnextsym v_4$). After examining a case like this, one would need to write an update formula to realize it.

The proof of~\cref{lemma:eqsubstr} looks through all the cases and produces a $\ucq$ update formula for each. These subformulae are joined together by disjunction to give us an update formula $\updateformula{\equalsubstr}{\unknownupdate}{u}$ which is in $\dynucq$, and hence we have proven that we can maintain the equal substring relation in $\dyncq$. See the appendix for the proof.

\begin{figure}
\tikzset{every picture/.style={line width=0.75pt}} 
\begin{tikzpicture}[x=0.75pt,y=0.75pt,yscale=-1,xscale=1]
\draw    (70.1,50) -- (70.1,70) ;
\draw    (130.1,50) -- (130.1,70) ;
\draw    (170.43,50) -- (170.43,70) ;
\draw    (230.1,50) -- (230.1,70) ;
\draw    (150.1,50) -- (150.1,70) ;
\draw    (360.1,50) -- (360.1,70) ;
\draw    (420.1,50) -- (420.1,70) ;
\draw    (460.43,50) -- (460.43,70) ;
\draw    (520.1,50) -- (520.1,70) ;
\draw    (440.1,50) -- (440.1,70) ;
\draw    (30,60) -- (550,60) ;
\draw (150.33,41) node    {$u$};
\draw (150.07,76) node  [color={rgb, 255:red, 208; green, 2; blue, 27 }  ,opacity=1 ] [align=left] {a};
\draw (100.1,76) node    {$w_{1}$};
\draw (200.17,76) node    {$w_{2}$};
\draw (71,41) node    {$x_{1}$};
\draw (229.67,41) node    {$y_{1}$};
\draw (440.33,41) node    {$v$};
\draw (439.97,76) node   [align=left] {a};
\draw (390,76) node    {$w_{1}$};
\draw (490.17,76) node    {$w_{2}$};
\draw (361,41) node    {$x_{2}$};
\draw (519.67,41) node    {$y_{2}$};
\draw (130.2,41) node    {$v_{1}$};
\draw (170.2,41) node    {$v_{2}$};
\draw (420.2,41) node    {$v_{3}$};
\draw (459.8,41) node    {$v_{4}$};
\end{tikzpicture}
\caption{Word after the insertion of the symbol $a$ at node $u$.\label{fig:word}.}
\end{figure}
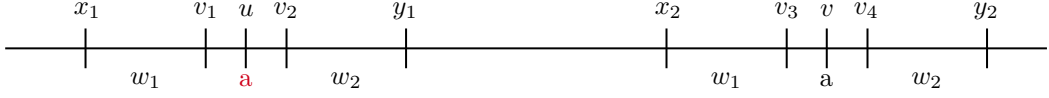

\cref{lemma:eqsubstr} is a central part of the proof of our main result, and some may consider maintaining this relation also to be the most technical aspect of the present paper. This relation will be the main feature of a construction to maintain so-called \emph{pattern languages}, which we then extend with regular constraints to maintain any relations selectable by~\splog.

Given a pattern $\alpha \in (\Sigma \cup \Xi)^+$, we define the \emph{non-erasing} language it generates as $\nonerasing{\lang}(\alpha) \df \{ \sigma(\alpha) \mid \sigma \colon (\Sigma \cup \Xi)^+ \rightarrow \Sigma^+ \text{ where $\sigma$ is a substitution}  \}$. Given the same pattern~$\alpha$, we have $\erasing{\lang}(\alpha) \df \{ \sigma(\alpha) \mid \sigma \colon (\Sigma \cup \Xi)^+ \rightarrow \Sigma^* \text{ where $\sigma$ is a substitution}  \}$ which is the \emph{erasing} language~$\alpha$ generates. Pattern languages are not only used as a part of word equations but also as language generators (see \cite{fre:doc} for more details, in particular regarding their relation to document spanners).

\begin{example}
Consider $\alpha \df axxb$ where $a,b \in \Sigma$ and $x \in \Xi$. Then $ab \in \erasing\lang(\alpha)$ with $\sigma(x)=\emptyword$, but $ab \notin \nonerasing\lang(\alpha)$. We can also see that $ ababab \in \nonerasing\lang(\alpha)$ and $ababab \in \erasing\lang(\alpha)$ using $\sigma(x) = ba$.
\end{example}

We take the definition of maintaining a language from \cite{gel:dyn}. We can \emph{maintain a language} $L$ if a dynamic program maintains a 0-ary relation which is true if and only if $\dynword(\wordstruc) \in L$.

\begin{lemma}
\label{prop:patterns}
Every non-erasing pattern language can be maintained in \dyncq.
\end{lemma}
\begin{proof}
To prove this lemma, we give a way to symbolically construct an update formula to maintain a 0-ary relation $\mathcal{P}$ which updates to true if and only if $w' \in \nonerasing{\lang}(\alpha)$ for any specified $\alpha \in (\Sigma \cup \Xi)^+$. Let $|\alpha|$ be the length of the pattern $\alpha$. Let $\alpha_i$ denote the $i^{th}$ symbol (from $\Xi$ or $\Sigma$) of the pattern $\alpha$ where $1 \leq i \leq |\alpha|$. We give the construction in Algorithm~\ref{algo:patLangUpdate}.

\SetKwInput{KwInput}{Input}
\SetKwInput{KwOutput}{Output}

\begin{algorithm}
\textbf{Input:} A pattern $\alpha \in (\Sigma \cup X)^+$. \; \\
\textbf{Output:} Update formulas $\updateformula{\patternquery}{\absins{\zeta}}{u}$ and $\updateformula{\patternquery}{\absreset}{u}$. \; \\
If $\alpha_1 \in \Sigma$ then $\omega_1 := R_{\alpha_1}(t_1) \land \firstrel'(t_1) $; \; \\
If $\alpha_1 \in \Xi$ then $\omega_1 := (x_1 \leq t_1) \land \firstrel'(x_1) $; \; \\
  \For{$i:=2$ \KwTo $|\alpha|$ }{
    	\If{$\alpha_i \in \Sigma$}{
   		$\omega_i := \symbolrel{\alpha_i}(t_i) \land \nextrelation'(t_{i-1},t_i) \land \omega_{i-1} $; \; 
   	}
	\If{$\alpha_i \in \Xi $}{
		\eIf{there exists $j \in \mathbb{N}$ where $j<i$ such that $\alpha_i = \alpha_j$}{
			$j_{max} :=$ Largest $j$ value such that $j<i$ and $\alpha_i = \alpha_j$; \; \\
			$\omega_i := \nextrelation'(t_{i-1}, x_i) \land (x_i \leq t_i) \land \equalsubstr'(x_{j_{max}}, t_{j_{max}}, x_i, t_i)  \land \omega_{i-1}$; \;
		}{
			$\omega_i := \nextrelation'(t_{i-1}, x_i) \land (x_i \leq t_i) \land \omega_{i-1}$; \;
		}
   	}
    }
$\omega := \big( \omega_{|\alpha|} \land \lastrel'(t_{|\alpha|}) \big)$; \; \\
For every occurrence of some  $t_i$ in $\omega$, where $i \leq |\alpha|$, add $\exists t_i$ to the front of $\omega$; \; \\
For every occurrence of some $x_i$ in $\omega$ add $\exists x_i$ to the front of $\omega$; \; \\
$\phi_{\absins{\zeta}}^{\patternquery}(u) := \omega$; \;
$\phi_{\absreset}^{\patternquery}(u) := \omega$; \;
\caption{Pattern Language Update Formula Construction.}\label{algo:patLangUpdate}
\end{algorithm}

Note that occurrences of $\nextrelation'$ and $\equalsubstr'$ in Algorithm~\ref{algo:patLangUpdate} are the relations correct \emph{after} the update. To achieve this, we can replace occurrences of $\nextrelation'(\dots)$ with $\updateformula{\nextrelation}{\unknownupdate}{\dots}$, where $\unknownupdate$ is the update for which the update formula of $\mathcal{P}$ is being constructed. The equivalent is done for $\equalsubstr$.
\end{proof}

\begin{example}
Let $\alpha \df axbx$ be a pattern such that $a,b \in \Sigma$ and $x \in \Xi$. As stated, we wish to maintain a 0-ary relation $\mathcal{P}$ such that $\mathcal{P}$ is true if and only if $w' \in \nonerasing\lang(\alpha)$ where $w'$ is our word after some update. 
\begin{itemize}
\item $\alpha_1 = a$: therefore $\alpha_1\in \Sigma$ and hence we have $\omega_1 \df R_{a}(t_1) \land \firstrel'(t_1)$.
\item $\alpha_2 = x$: therefore $\alpha_2\in\Xi$ therefore we have $\omega_2 \df \nextrelation'(t_1, x_2) \land (x_2 \leq t_2) \land \omega_1$.
\item $\alpha_3 = b$: therefore $\alpha_3 \in \Sigma$ and hence we have $\omega_3 \df R_b(t_3) \land \nextrelation'(t_2,t_3) \land \omega_2$.
\item $\alpha_4 = x$ and $\alpha_4 = \alpha_2$: therefore $\omega_4 \df \nextrelation'(t_3,x_4) \land (x_4\leq t_4) \land \equalsubstr'(x_2,t_2,x_4,t_4) \land \omega_3$.
\end{itemize}

We rearrange the atoms in $\omega$ to help with readability, giving us:
\begin{multline*}
\omega \df  \firstrel'(t_1) \land R_{a}(t_1) \land \nextrelation'(t_1, x_2) \land (x_2 \leq t_2) \land \nextrelation'(t_2,t_3)  \land R_b(t_3)  \\
\land \nextrelation'(t_3,x_4) \land (x_4\leq t_4) \land \equalsubstr'(x_2,t_2,x_4,t_4) \land \lastrel'(t_4).
\end{multline*}

Hence $\updateformula{\mathcal{P}}{\unknownupdate}{u} \df \exists t_1,t_2,t_3,t_4,x_2,x_4 \colon (\omega)$ which holds for a word-structure of the form:
\begin{center}
	\begin{tabular}{cccccccccc}
		$\dots$ & $ t_1 $&$ \mathbf{x_2} $&$ \mathbf{\dots} $&$ \mathbf{t_2} $&$ t_3 $&$ \mathbf{x_4} $&$ \mathbf{\dots} $&$ \mathbf{t_4} $ & $\dots$ \\
		$\emptyword$ & $a$ & & $\dots$ & & $b$ &  & $\dots$ & & $\emptyword$ 
	\end{tabular}
\end{center}

We have that $x_2,t_2,x_4,t_4$ are in bold to demonstrate the fact that it must be that $w'[x_2,t_2] = w'[x_4,t_4]$ for $\updateformula{\mathcal{P}}{\unknownupdate}{u}$ to hold. Note that $t_1$ may not be $<-minimal$ and $t_4$ may not be $<-maximal$, but because $\firstrel'(t_1)$ and $\lastrel'(t_4)$ must hold, $t_1$ and $t_4$ are the first and last symbol-elements respectively.
\end{example}

One side effect of~\cref{prop:patterns} is that we get the dynamic complexity upper bounds of a class of languages, the pattern languages. Pattern languages were not looked at in~\cite{gel:dyn} and hence this result extends what is known about the dynamic complexity of formal languages.

\begin{corollary}
\label{cor:erasing}
Every erasing pattern language can be maintained in \dyncq.
\end{corollary}
\begin{proof}

From Jiang et al. \cite{jiang:pat} it is known that every erasing pattern language is the finite union of non-erasing pattern languages. Therefore, we can create 0-ary relations for each non-erasing pattern language and join them with a disjunction. There is the case where $\emptyword \in \erasing{\lang}(\alpha)$ which we can deal with using the following: $\exists x \colon ( \firstrel(x) \land (x \logeq \$) )$. We can do this because $\firstrel = \{\$ \}$ whenever $w=\emptyword$.
\end{proof}

Since we are able to maintain any erasing pattern language in \dyncq, we can extend this result to word-equations in \splog-formulas. Using this along with the fact that regular languages can be maintained in \dynprop, we can conclude the following: 

\begin{lemma}\label{lem:splogDynCQ}
Any relation selectable in \splog can be maintained in \dyncq.
\end{lemma}
\begin{proof}
We prove this lemma using structural induction with the recursive definition of a \splog formula, given in Definition \ref{def:splog}.

\textbf{B1.} $(\mv \weqeq \eta_R)$ for every $\eta_R \in  (\Xi\cup\Sigma)^*$: Since we are assuming that $\sigma(\mv) \in \Sigma^*$ and that $\eta_R$ does not contain $\mv$, we have that $\mv \weqeq \eta_R$ is equivalent to $\sigma(\mv) \in \erasing{\lang}(\eta_R)$. We have proven in~\cref{cor:erasing}, that we can maintain a 0-ary relation which is true if and only if, given some pattern $\alpha \in (\Xi\cup\Sigma)^*$, the word structure is currently a member of~$\erasing{\lang} (\alpha)$. According to the construction which we gave in Lemma \ref{prop:patterns}, given a variable $x \in \Xi$, where~$x = \alpha_{i}$, we have two variables $x_i, t_i \in \worddomain$ such that the word $w[x_i,t_i]$ represents $\sigma(x)$ for some substitution $\sigma$. Removing the existential quantifiers for $x_i$ and $t_i$ allows us to maintain the relation defined by $\alpha$.

\textbf{R1.} $(\psi_1 \land \psi_2)$ for all $\psi_1, \psi_2 \in \splog(\mv)$: Under the assumption that we have update formulas $\updateformula{\psi_1}{\unknownupdate}{u;\vec{v_1}}$ and $\updateformula{\psi_2}{\unknownupdate}{u;\vec{v_2}}$ for \splog formulas~$\psi_1$ and~$\psi_2$ respectively, the update formula for $\updateformula{\psi_1 \land \psi_2}{\unknownupdate}{u; \vec{v_1}\cup\vec{v_2} }$ is $\updateformula{\psi_1}{\unknownupdate}{u;\vec{v_1}} \land \updateformula{\psi_2}{\unknownupdate}{u;\vec{v_2}}$.

\textbf{R2.} $(\psi_1 \lor \psi_2)$ for all  $\psi_1, \psi_2 \in \splog(\mv)$ with $\fvar(\psi_1) = \fvar(\psi_2)$: Assuming we have update formulas $\updateformula{\psi_1}{\unknownupdate}{u;\vec{v}}$ and $\updateformula{\psi_2}{\unknownupdate}{u;\vec{v}}$ for \splog formulas~$\psi_1$ and~$\psi_2$ respectively, the update formula for $\updateformula{(\psi_1 \lor \psi_2)}{\unknownupdate}{u; \vec{v}}$ is $\updateformula{\psi_1}{\unknownupdate}{u;\vec{v}} \lor \updateformula{\psi_2}{\unknownupdate}{u;\vec{v}}$.

\textbf{R3.} $ \exists x \colon \psi$ for all $\psi \in \splog(\mv)$ and $x \in \fvar(\psi) \setminus \{ \mv \}$: If a variable $x \in \Xi$ is existentially quantified within the \splog formula, then we existentially quantify the variables $x_i,t_i \in \worddomain$ where $w[x_i,t_i]$ represents $\sigma(x)$ for some substitution $\sigma$.

\textbf{R4.} $(\psi \land \constr_A(x))$ for every $\psi \in \splog(\mv)$, every $x \in \fvar(\psi)$, and every NFA~$A$: let $A \df (Q, \delta, s, F)$ be an NFA. We have that $Q$ is a finite set of states, $\delta \colon Q \times \Sigma \rightarrow Q$ is the transition function, $s$ is the initial state and $F \subseteq Q$ is the set of accepting states. We denote the reflexive and transitive closure of $\delta$ as $\delta^* \colon Q \times \Sigma^* \rightarrow Q$. For regular constraints, we maintain the relation $R_{A} := \{ (i,j) \in \worddomain^2 \mid w[i,j] \in \mathcal{L}(A)  \}$

From Proposition~3.3 in Gelade, Marquardt, and Schwentick~\cite{gel:dyn}, we know that the following relations can be maintained in \dynprop, and from \cite{zeu:small} (Theorem~3.1.5, part~b) we know that \dynprop is a strict subclass of \dyncq. Hence we can maintain the following in \dyncq:
\begin{align*}
R_{p,q} \df & \{ (i,j) \in \worddomain^2 \mid i<j \text{ and }  \delta^*(p,w[i+1,j-1])  = q \}, \\
I_{q} \df & \{ j \in \worddomain \mid \delta^*(s,w[1,j-1]) = q \}, \\
F_{p} \df & \{ i \in \worddomain \mid \delta^*(p,w[i+1,n] \in F \}.
\end{align*}

Where $p,q \in Q$. We also know, from \cite{gel:dyn}, that we can maintain the 0-ary relation $\mathsf{ACC}$, which is true if and only if $w' \in \lang(A)$.

We maintain $R_{A}$ with $\updateformula{R_{A}}{\unknownupdate}{u;x,y} \df \psi_1^{R_{A}} \lor \psi_2^{R_{A}} \lor \psi_3^{R_{A}} \lor \psi_4^{R_{A}}$ where each $\psi_i^{R_{A}}$ is a subformula which we now define for separate cases. Note that~$R'(\vec{x})$ is shorthand for $\updateformula{R}{\unknownupdate}{u;\vec{x}}$. We define $\psi_1^{R_{A}} $ as
\[
\psi_1^{R_{A}} \df \exists x_2,y_2 \colon \big( \nextrelation'(x_2,x) \land \nextrelation'(y,y_1) \land \bigvee\limits_{f \in F}(R_{s,f}'(x_2,y_2) \big).
\]

Since $R_{p,q}(x,y)$ refers to the substring from position $x+1$ to $y-1$, and we wish to examine the string from position $x$ to $y$, we look at $R_{s,f}'(x_2,y_2)$ where $x_2 \newnextsym x$ and $y \newnextsym y_2$. If it is indeed the case that $x_2 \newnextsym x$ and $y \newnextsym y_2$ then $w'[x_2 +1, y_2 -1] = w[x,y]$. Therefore $R_{s,f}'(x_2,y_2)$, for $f \in F$, is true for such $x_2$ and $y_2$ if and only if $\delta^*(s, w[x,y]) \in F$ which is the desired behavior for this case. Note that $\psi_1^{R_{A}}$ fails if there doesn't exist $x_2$ such that $x_2 \newnextsym x$ or there doesn't exists $y_2$ such that $y \newnextsym y_2$. This is dealt with using $\psi_2^{R_{A}}, \psi_3^{R_{A}}$ and $\psi_4^{R_{A}}$, which we explore next.

If $\lastrel'(y)$ then $w'[x,y]=w'[x,n]$ where $n = |\worddomain|$. Therefore, we can use $F_s'(x_2)$ for some $x_2 \in \worddomain$ where $x_2 \newnextsym x$ and $s$ is the initial state of the NFA, to see whether $\delta^*(s, w'[x,n]) \in F$ and hence whether $\delta^*(s, w'[x,y]) \in F$. To realize this behavior, we define $\psi_2^{R_{A}} $ as
\[
\psi_2^{R_{A}} \df \exists x_2 \colon \big( \nextrelation'(x_2,x) \land \lastrel'(y) \land F_s'(x_2) \big).
\]

If $\firstrel'(x)$ then $w'[1,y] = w'[x,y]$. Therefore, we can use $I_f'(y_2)$ for some $y_2 \in \worddomain$ where $y \newnextsym y_2$ and $f \in F$, to see whether $\delta^*(s,w'[1,y]) \in F$ and hence whether $\delta^*(s,w'[x,y]) \in F$. To realize this behavior, we define $\psi_3^{R_{A}} $ as
\[
\psi_3^{R_{A}} \df \exists y_2 \colon \big( \nextrelation'(y,y_2) \land \firstrel'(x) \land \bigvee\limits_{f \in F}(I_f'(y_2)) \big).
\]

If $\firstrel'(x)$ and $\lastrel'(y)$ then $w'[x,y]=w'$ and therefore it follows that $w'[x,y] \in \lang(A)$ if and only if $w' \in \lang(A)$. We only need to see if $\mathsf{ACC'}$ is true for this case. We realize this behavior by defining $\psi_4^{R_{A}}$ as
\[
\psi_4^{R_{A}} \df \firstrel'(x) \land \lastrel'(y) \land \mathsf{ACC'}.
\]

To simulate $(\psi \land \constr_A(x))$ for every $\psi \in \splog(\mv)$, every $x \in \fvar(\psi)$, and every NFA~$A$ within \dyncq, we do the following; let $\updateformula{\psi}{\unknownupdate}{u;\vec{v}}$ be an update formula for $\psi \in \splog$ and since for some $\sigma(x)$, where $x \in \fvar(\psi)$, has $x_i,t_i \in \worddomain$ associated with it, we can use $\updateformula{\psi}{\unknownupdate}{u;\vec{v}} \land \updateformula{R_A}{\unknownupdate}{u;x_i,t_i}$ which is true if and only if $w'[x_i,t_i] \in \lang(A)$.
\end{proof}

Most of the work for this proof follows from Lemma \ref{prop:patterns} and Corollary \ref{cor:erasing}. Extra work is done in order to simulate regular constraints, although this follows on from the fact that \dynprop maintains the regular languages \cite{gel:dyn}. 

\begin{theorem}\label{splogindyncq}
Core spanners can be maintained in \dyncq.
\end{theorem}
\begin{proof}
Although maintaining the \splog relation that realizes a spanner is not the same as maintaining the spanner relation as defined in~\cref{sec:dyn}, the changes we need to make are trivial. Let $P$ be a spanner and let $\psi_P$ be a \splog formula that realizes $P$. We know that $\fvar(\psi_P) = \{ x_p, x_c \mid x \in \SVars{P} \}$, and for every $x \in \SVars{P}$ where $\spn{i,j} \df \mu(x)$, we have both $\sigma(x_p) = w_{\spn{1,i}}$ and $\sigma(x_c) = w_{\spn{i,j}}$. Let $R^P$ be a relation that maintains the spanner~$P$. The only difference between update formulas that maintain $P$ and update formulas that maintain the relation \splog selects which realizes $P$ is that the two elements $x_p^o, x_p^c \in \worddomain$ that are used to represent the \splog variable $x_p \in \Xi$ are existentially quantified whereas the two variables $x_c^o, x_c^c \in \worddomain$ which represent $x_c \in \Xi$ are not. 
\end{proof}

\cref{splogindyncq} shows us that \dyncq is at least as expressive as \splog. We will use this along with Proposition \ref{eqlen} to show that \dyncq is more expressive than core spanners. Given that we can maintain any relation selectable in \splog using \dyncq, it is no big surprise that adding negation allows us to maintain \splogneg in \dynfo.

\begin{lemma}\label{splogneg}
Any relation selectable in \splogneg can be maintained in \dynfo.
\end{lemma}
\begin{proof}
Let $\psi \in \splog(\mv)$ and let $R^{\psi}$ be the relation maintaining $\psi$ where the update formulas for $R^{\psi}$ are in \cq. The extra recursive rule allowing for $(\neg\psi) \in \splogneg(\mv)$ can be maintained by $\updateformula{R^{\neg\psi}}{\unknownupdate}{u;\vec{x}} = \neg \updateformula{R^{\psi}}{\unknownupdate}{u;\vec{x}}$.
\end{proof}

As with Theorem \ref{splogindyncq}, we can use the result from Lemma \ref{splogneg} along with Corollary \ref{pow2} to show that \dynfo is more expressive than \splogneg.

\begin{theorem}\label{genCoreInDynFO}
Generalized core spanners can be maintained in \dynfo.
\end{theorem}
Since \splogneg captures the generalized core spanners, it follows from~\cref{splogneg} that any generalized core spanner can be maintained in \dynfo. In Section 4 we show that \dynfo is more expressive than \splogneg, it therefore follows that \dynfo is more expressive than generalized core spanners.

\section{Relations in \splog and \dyncq}\label{sec:rel}
In this section, we examine the comparative expressive power of \splog and \dyncq.  
Recall that we defined the notion of \splog-selectable relations at the end of \cref{sec:splog}. We now define an analogous concept for \dyncq. For a relation $R\subseteq (\Sigma^*)^k$, we define the corresponding relation in the dynamic setting $\bar{R}$ as the $2k$-ary relation of all $(x_1,y_1,\ldots,x_k,y_k)\in \worddomain^{2k}$ such that $(w[x_1,y_1], \ldots, w[x_k,y_k])\in R$. We say that $R$ is selectable in \dyncq if $\bar{R}$ can be maintained in \dyncq.

For example, the equal length relation is defined as $R_{\mathsf{len}} \df \{ (w_1,w_2) \mid |w_1| = |w_2| \}$. From Fagin et al. \cite{fag:spa} it is known that this relation is not selectable with core spanners. This relation in the dynamic setting is $\bar{R}_{\mathsf{len}} = \{ (u_1,u_2,v_1,v_1) \in \worddomain^4 \mid |w[u_1,u_2]| = |w[v_1,v_2]| \}$.

\begin{proposition}\label{eqlen}
	The equal length relation is selectable in \dyncq.
\end{proposition}
\begin{proof}
	To maintain the equal length relation, we take the update formulas from \cref{lemma:eqsubstr} and omit any atoms relating to the symbol of an element of the domain $\worddomain$. We also remove the constraint that the first subword must appear before the second. We also use $\bar{R}_{\mathsf{len}}$ in any update formula, rather than $R_{\mathsf{eq}}$. The only exception to omitting all atoms relating to the symbol of an element, is to ensure that $w[u_1] \neq \emptyword$, $w[u_2] \neq \emptyword$, $w[v_1] \neq \emptyword$, and $w[v_2] \neq \emptyword$.
\end{proof}

While this allows us to separate the languages that are definable in \splog from the ones that can be maintained in \dyncq, we consider the following more wide-ranging example:
\begin{lemma}\label{pow2}
	The language  $\{w\in\Sigma^* \mid |w| = 2^n, n\geq 0\}$  is maintainable in \dyncq.
\end{lemma}
\begin{proof}
Let $P$ be a 2-ary relation such that $P(x,y)$ holds if and only if $|w[x,y]| = 2^n$ for some $n \in \mathbb{N}$. This can be maintained by having that $P(x,y)$ holds if $|w[x,y]| = 1$ or if there exists $z_1, z_2 \in \worddomain$ such that $P(x,z_1)$, $P(z_2,y)$, $\nextrelation'(z_1,z_2)$ and that $\bar{R}_{\mathsf{len}}(x,z_1,z_2,y)$. If we assume that $|w[x,z_1]| = 2^n$ for some $n \in \mathbb{N}$, which we do because we have the base case of $w[x,y]=a$, and that $|w[x,z_1]|=|w[z_2,y]|$, then it follows that if $\nextrelation'(z_1,z_2)$ then $w[x,y] = w[x,z_1] \cdot w[z_2,y]$ and therefore $|w[x,y]| = 2 |w[x,z_1]|$ and hence $|w[x,y]| = 2^{n+1}$. We then have that $|w| = 2^n$ if $\exists x,y \colon \big( \firstrel'(x) \land \lastrel'(y) \land P'(x,y) \big)$.
\end{proof}

For every choice of  $\Sigma$, this language is not expressible in \splogneg (and, hence, not in \splog). This is easily seen by considering the case that $\Sigma$ is unary\footnote{Larger alphabets then follow by observing that the class of $\splogneg$-languages is trivially closed under intersection with regular languages.}. As shown in~\cite{fre:doc} for core spanners and then in~\cite{pet:rec} for generalized core spanners, both classes collapse to exactly the class of regular languages if $|\Sigma|=1$. As the language of all words $a^{2^n}$ is not regular, this shows that even \dyncq can define languages that are not expressible in \splogneg. 

Combining this with Theorem~\ref{splogindyncq} and Theorem~\ref{genCoreInDynFO}, we respectively conclude that \dyncq is strictly more expressive than core spanners and that \dynfo is strictly more expressive than generalized core spanners. 

As explained in Section~6 of~\cite{fre:splog}, there are few inexpressibility results for \splog that generalize to non-unary alphabets (and basically none for \splogneg), apart from straightforward complexity observations that are not particularly illuminating.  
Nonetheless, Proposition~6.7 in~\cite{fre:splog} establishes that none of the following relations is \splog-selectable:
\begin{proposition}\label{prop:beyondsplog}
The following relations are \dyncq-selectable but not \splog-selectable:
\begin{align*}
	R_{\mathsf{num}(a)} &\df \{ (w_1,w_2) \mid |w_1|_a = |w_2|_a \} \text{ for } a \in \Sigma, \\
	R_{\mathsf{perm}} &\df \{ (w_1,w_2) \mid |w_1|_a = |w_2|_a \text{ for all } a \in \Sigma\},\\
	R_{\mathsf{rev}} &\df \{ (w_1,w_2) \mid w_2=w_1^R \}, \text{ where $w_1^R$ is the reversal of $w_1$,}\\
	R_{<} &\df \{ (w_1,w_2)\mid |w_1| < |w_2| \},\\
	R_{\mathsf{scatt}} &\df \{ (w_1,w_2) \mid w_1 \text{ is a scattered subword of } w_2 \},
\end{align*}
where $w_1$ is a scattered subword of $w_2$ if, for some $n\geq 1$, there exist $s_1,\ldots,s_n,\bar{s}_0,\ldots,\bar{s}_n\in\Sigma^*$ such that $w_1=s_1\cdots s_n$ and $w_2= \bar{s}_0 s_1 \bar{s}_1 \cdots s_n \bar{s}_n$.
\end{proposition}
\begin{proof}
The relations $R_{\mathsf{scatt}}$, $R_{\mathsf{num}(a)}$, and $R_{\mathsf{rev}}$ have case distinctions equivalent to the proof of~\cref{lemma:eqsubstr}, therefore we give the overarching idea of the proof but without exploring every case. See \cite{frey2019dynfo} for a full proof of~\cref{lemma:eqsubstr}.
\subparagraph*{Maintaining $R_{\mathsf{scatt}}$:}

For insertion, we give three steps for this proof; inheritance, base case, and an inductive step.

We have that if $w[u_1,u_2]$ is a scattered subword of $w[v_1,v_2]$ and $u$ is outside of the interval $[u_1,u_2]$, then $w'[u_1,u_2]$ remains a scattered subword of $w'[v_1,v_2]$ and therefore $R'_{\mathsf{scatt}}(u_1,u_2,v_1,v_2)$ should hold. We call this step inheritance.

The base case is that given the update $\ins{\zeta}{u}$ for some $u \in \worddomain$, if there exists $v \in \worddomain$ such that $v_1 \leq v \leq v_2$ and $w(v) = w(u) = \zeta$, then it follows that $w(u)$ is a scattered subword of $w[v_1,v_2]$ and therefore $R'_{\mathsf{scatt}}(u,u,v_1,v_2)$ should hold.

For the inductive step, given that we have some update $\ins{\zeta}{u}$, if $w[u_1,x_1]$ is a scattered subword of $w[v_1,x_2]$ and $w[x_3,u_2]$ is a scattered subword of $w[x_4,v_2]$, it follows that $w[u_1,u_2]$ is a scattered subword of $w[v_1,v_2]$ if $x_1\newnextsym u\newnextsym x_3$ and $w(u)$ is a scattered subword of $w[x_2,x_4]$. Deletion is dealt with analogously, although without the base case. 
\subparagraph*{Maintaining $R_{\mathsf{num}(a)}$:}

We again give three steps; inheritance, the base case(s), and an inductive step.

We have that if $|w[u_1,u_2]|_a = |w[v_1,v_2]|_a$ and $u$ is outside of the interval $[u_1,u_2]$, then $|w'[u_1,u_2]|_a = |w'[v_1,v_2]|_a$ and therefore $R'_{\mathsf{num}(a)}(u_1,u_2,v_1,v_2)$ should hold. We call this step inheritance. We have that $(u_1,u_2,v_1,v_2)$ is not inherited if $u \in [u_1,u_2]$ or $u \in [v_1,v_2]$, but this should be dealt with by the inductive step. 

To maintain $R_{\mathsf{num}(a)}$, we have two base cases. Given the update $\ins{a}{u}$, we have that $|w'(u)|_{a} = |w'(v)|_{a}$ if $w'(v) = a$.

For the inductive step, we have that if $|w[u_1,x_1]|_a = |w[v_1,x_2]|_a$ and $|w(u)|_a = |w(v)|_a$ and $|w[x_3,u_2]|_a = |w[x_4,v_2]|_a$ where $x_1 \newnextsym u \newnextsym x_3$ and $x_2 \newnextsym v \newnextsym x_4$, then $|w'[u_1,u_2]|_a = |w'[v_1,v_2]|_a$. Dealing with deletion is analogous to insertion but without the base case.

\subparagraph*{Maintaining $R_{\mathsf{rev}}$:}

We can maintain this with a simple variation of the update formula which maintains $\equalsubstr$. Firstly, we remove the constraint that the first subword must appear before the second. Then, whenever $\equalsubstr(\cdot)$ is used as a subformula, one would need to use $R_{\mathsf{rev}}(\cdot)$ instead. The more involved aspect of altering the update formulas would be to reverse the ordering of certain indices. Informally, check $y \nextsym x$ instead of $x \nextsym y$ where necessary.
\subparagraph*{Maintaining $R_{\mathsf{perm}}$:} $\updateformula{R_{\mathsf{perm}}}{\unknownupdate}{u;u_1,u_2,v_1,v_2} \df \bigwedge\limits_{\zeta \in \Sigma} \big( \updateformula{R_{\mathsf{num}(\zeta)}}{\unknownupdate}{u;u_1,u_2,v_1,v_2} \big). $
\subparagraph*{Maintaining $R_{<}$:}

\begin{multline*}
\updateformula{R_{<}}{\unknownupdate}{u;u_1,u_2,v_1,v_2} \df  \exists x_1 \exists x_2 \colon \big( R_{\mathsf{len}}(u_1,u_2,x_1,x_2)  \\
\land (x_1<v_1) \land (v_1\leq v_2) \land (v_2<x_2)\big).\qedhere 
\end{multline*}
\end{proof}

By Lemma~5.1 in~\cite{fre:splog}, a $k$-ary relation $R$ is \splog-selectable if and only there is some \splog-formula $\varphi(\mv;x_1,\ldots,x_k)$ such that for all $\sigma$ that satisfy $\sigma(x_i)\subword \sigma(\mv)$ for all $i\in[k]$, we have $\sigma\models\varphi$ if and only if $(\sigma(x_1),\ldots,\sigma(x_k))\in R$. One can show with little effort that relations like string inequality, the substring relation, or equality modulo a bounded Levenshtein-distance are all \splog-selectable (see Section~5.1 of~\cite{fre:splog}). By \cref{lem:splogDynCQ}, we can directly use these relations in constructions for \dyncq-definable languages and \dyncq-selectable relations.
\begin{example}
For $k\geq1$ and $w_1,w_2\in\Sigma^*$, we say that $w_1$ is a $k$-scattered subword of~$w_2$ if there exist $s_1,\ldots,s_k,\bar{s}_0,\ldots,\bar{s}_k\in\Sigma^*$ such that $w_1=s_1\cdots s_k$ and $w_2= \bar{s}_0 s_1 \bar{s}_1 \cdots s_k \bar{s}_k$. This relation is \splog-selectable\footnote{Unlike a relation for unbounded scattered subword.}, as demonstrated by the following \splog-formula which uses syntactic sugar from Section~5.1 of~\cite{fre:splog}:
\[
\varphi(\mv;w_1,w_2)\df
\exists s_1,\ldots,s_k,\bar{s}_0,\ldots,\bar{s}_k:
\Bigl( (w_1 \weqeq  s_1\cdots s_k)
\land (w_2 \weqeq  \bar{s}_0 s_1 \bar{s}_1\cdots s_k \bar{s}_k)
\Bigr).
\]
\end{example}
Although one could show directly that the $k$-scattered subword relation is \dyncq-selectable, using \splog and \cref{lem:splogDynCQ} can avoid hand-waving.

We can even generalize this approach beyond \splog.
In the proof of \cref{lem:splogDynCQ}, we use the fact the every regular language is in \dyncq to maintain regular constraints for \splog. 
Analogously, we can extend \splog with relation symbols for any \dyncq-sectable relation and use the resulting logic for~\dyncq. Of course, all this applies to \splogneg and \dynfo.

\section{Conclusions}\label{sec:conc}
From a document spanner point of view, the present paper establishes upper bounds for maintaining the three most commonly examined classes of document spanners, namely \dynprop for regular spanners, \dyncq for core spanners, and \dynfo for generalized core spanners. While the bounds for regular spanners and generalized core spanners are what one might expect from related work, the \dyncq-bound for core spanners might be considered surprising low (keeping in mind, of course, that it is still open whether \dyncq is less expressive than \dynfo).

By analyzing the proof of \cref{lem:splogDynCQ}, the central construction of this main result, it seems that the most important part of maintaining core spanners is updating the string equality relation and the regular constraints. 
One big question for future work is whether this might have any practical use for the evaluation of core spanners. 
Although some may consider this unlikely, there is at least some possibility that some techniques might be useful.

In the present paper, we only examine updates that affect single letters. At least as far as the main result is concerned, it should be possible to generalize this to cut and paste operations, as they are commonly found in text editors. These other operations beyond single letters are promising directions for further work.

From a dynamic complexity point of view, \cref{sec:rel} describes how \splog can be used as a convenient tool that allows shorter proofs that languages can be maintained in \dyncq. One consequence of this is that a large class of regular expressions with backreference operators (see Section~5.3 of~\cite{fre:splog}) are in fact \dyncq-languages.

\newpage
\bibliography{bibliography}
\newpage
\appendix
\section{Proofs for \cref{sec:main}}
\subsection{Proof of \cref{lemma:next}}
We first observe the following helpful result:
\begin{lemma}\label{obs:nextins}
	Let $\unknownupdate = \ins{\zeta}{u}$ and let $x \nextsym y$ for $x,y \in \worddomain$. We have that $x \not\newnextsym y$ if and only if $x<u<y$.
\end{lemma}
\begin{proof}
	Let $\zeta \in \Sigma$, if we perform the update $\ins{\zeta}{u}$ on $\wordstruc$ where $x<u<y$ then it follows that there exists some $z$ such that $w'(z) \neq \emptyword$ and $x<z<y$. Therefore it cannot be the case that $x \newnextsym y$, so $x \not\newnextsym y$.
	
	If it is not the case that $x<u<y$ then it cannot be that there exists some $z \in \worddomain$ such that $x<z<y$ where $w'(z) \neq \emptyword$. Therefore, if $x \nextsym y$ and $\unknownupdate = \ins{\zeta}{u}$ then $x \not\newnextsym y$ if and only if $x<u<y$.
\end{proof}
\subsubsection{Actual proof of \cref{lemma:next}}
\begin{proof}
We first define the relations $\firstrel$ and $\lastrel$. These are unary relations which have the first and last symbol elements in a word structure respectively. Formally, we define them as $\firstrel := \{ x \in \worddomain \mid \position{w}{x}=1 \}$ and $\lastrel := \{ x \in \worddomain \mid \position{w}{x} = |w| \}$. Since $\position{w}{x}$ for any $x \in \worddomain$ is undefined when $w=\emptyword$, we use the following initialization $\firstrel := \{ \$ \}$ and $\lastrel := \{ 1 \}$. We also have that $\nextrelation$ is initialized to $\emptyset$.

We split this proof into two parts; one part for the insertion update and one part for the reset update.

\subparagraph*{Part 1 (insertion):}

To prove this part, we assume the relations $\nextrelation,\firstrel,\lastrel \in \auxstruc$ are correct for some arbitrary word-structure $\wordstruc$, and then prove that they are correctly updated for $\unknownupdate(\wordstruc)$, where $\unknownupdate = \ins{\zeta}{u}$. We now define the update formula for the $\nextrelation$ relation under $ins_{\zeta}$:
\[
\updateformula{\nextrelation}{\absins{\zeta}}{u;x,y} \df \bigvee\limits_{i=1}^{5} \big( \nextsubform{i} \big).
\]
where each $\nextsubform{i}$ is a UCQ subformula defined later. For readability, we denote the relation defined by $\{ (x,y) \in \worddomain^2 \mid \programstate \models\updateformula{\nextrelation}{\absins{\zeta}}{u;x,y} \}$ as $\nextrelation'$, where $\programstate \df (\wordstruc,\auxstruc)$ is the program state. We also do the analogous for $\firstrel$ and $\lastrel$.

\begin{description}
\item[Case 1.] $(x,y) \in \nextrelation$.
\end{description}

For this case, we refer back to Lemma \ref{obs:nextins}. From this lemma, we can see that if $x \nextsym y$ and $x<u<y$ then $x \not\newnextsym y$. It follows that if $(x,y) \in \nextrelation$ and $(x<u<y)$ then we should have $(x,y) \notin \nextrelation'$. We can also see from this lemma that if $x \nextsym y$ and $u \leq x$ or $y \leq u$ then $x \newnextsym y$ and therefore if $(x,y) \in \nextrelation$ and $(u \leq x) \lor (y \leq u)$ then $(x,y) \in \nextrelation'$. We can see that this behavior is realized with the following
\[
\nextsubform{1} \df \nextrelation(x,y) \land \big( (u \leq x) \lor (y \leq u) \big).
\]

\begin{description}
\item[Case 2.] $(x,y) \notin \nextrelation$ and $(x,y) \in \nextrelation'$.
\end{description}

We can see that if $(x,y) \notin \nextrelation$ and $u \neq x$ or $u \neq y$ then it must be that $(x,y) \notin \nextrelation'$. This is because either:
\begin{itemize}
\item $w(x) = \emptyword$ or $w(y) = \emptyword$ - this doesn't change if $u \neq x$ or $u \neq y$.
\item There exists some $v \in \worddomain$ such that $x<v<y$ and $w(v) \neq \emptyword$ - since we are looking at when $\unknownupdate = \ins{\zeta}{u}$, we still have such an element $v$.
\end{itemize}

Therefore, we will look at two cases; when $u = x$ and when $u = y$:

\begin{description}
\item[Case 2.1.] $u=x$.
\end{description}

We first look at when $\position{w'}{u} = 1$. We now define $\nextsubform{2}$:
\[
\nextsubform{2} \df (u \logeq x) \land \firstrel(y) \land (u < y).
\]
We will assume that $\nextsubform{2}$ evaluates to true and show that $x \newnextsym y$. For $\nextsubform{2}$ to be true, it must be that:
\begin{itemize}
\item $u = x$.
\item $\firstrel(y)$ - which is the case when $\position{w}{y} = 1$.
\item $(u<y)$.
\end{itemize}

Since $\position{w}{y} = 1$ and $u<y$ it follows that $\position{w'}{u} = 1$. Furthermore, we can see that because $u<y$ we have that $\position{w'}{y} = \position{w}{y}+1$. It follows that $\position{w'}{u} =1$ and $\position{w'}{y} =2$ and therefore $u \newnextsym y$. Since $u=x$ we have $x \newnextsym y$, hence this subformula has the correct behavior for this case when $\position{w'}{u} = 1$. But we are still yet to explore when $\position{w'}{u} \neq 1$. We now look at $\nextsubform{3}$:
\[
\nextsubform{3} \df (u \logeq x) \land \exists v \colon \big( \nextrelation(v,y) \land (v<u) \land (u<y) \big). 
\]
Assuming that $\nextsubform{3}$ evaluates to true, it must be that there exists some $v \in \worddomain$ such that:
\begin{itemize}
\item $u = x$.
\item $\nextrelation(v,y)$ - therefore $v \nextsym y$.
\item $v<u$ and $u<y$.
\end{itemize}

We know that $u = x$, therefore we can refer to $x$ as the element of the domain for which the symbol is being set. Since $v \nextsym y$ and $v<x<y$, it follows that $v \newnextsym x \newnextsym y$. Therefore we can see that $x \newnextsym y$ and $(x,y) \in \nextrelation'$, which is the correct behavior for $\nextsubform{3}$ in this case. 

\begin{description}
\item[Case 2.2.] $u=y$.
\end{description}

This case is analogous to Case 2.1. We have $\nextsubform{4}$ for when $\position{w'}{u}= |w'|$ and we have $\nextsubform{5}$ for when $\position{w'}{u} \neq |w'|$:
\begin{align*}
& \nextsubform{4} \df (u \logeq y) \land \lastrel(x) \land (u > x), \\
& \nextsubform{5}\df (u \logeq y) \land \exists v \colon \big( \nextrelation(x,v) \land (x < u) \land (u < v) \big).
\end{align*}

The intuition behind these subformulas is analogous to the reasoning stated for $\nextsubform{2}$ and $\nextsubform{3} $.

\begin{description}
\item[Case 3.] $(x,y) \notin \nextrelation$ and $(x,y) \notin \nextrelation'$.
\end{description}

This is the case where none of the subformulas evaluate to true, and therefore $\updateformula{\nextrelation}{\absins{\zeta}}{u;x,y}$ evaluates to false. Hence $(x,y) \notin \nextrelation'$.

We have proven for each case, the correctness of the update formula for $\nextrelation$ under insertion. We now prove the correctness of $\firstrel$ and $\lastrel$ by giving update formulas for them under the update $\unknownupdate = \ins{\zeta}{u}$:
\begin{align*}
& \updateformula{\firstrel}{\absins{\zeta}}{u;x} \df \big( \firstrel(x) \land (u>x) \big) \lor \exists y \colon \big( \firstrel(y) \land (u<y) \land (u \logeq x) \big), \\
& \updateformula{\lastrel}{\absins{\zeta}}{u;x} \df \big( \lastrel(x) \land (u<x) \big) \lor \exists y \colon \big( \lastrel(y) \land (u>y) \land (u \logeq x) \big).
\end{align*}

The intuition behind $\updateformula{\firstrel}{\absins{\zeta}}{u;x}$ is, if $u<x$ where $x$ is the first symbol element, then $u$ is the new first symbol element, otherwise $x$ remains the first symbol element. The intuition for $\updateformula{\lastrel}{\absins{\zeta}}{u;x}$ follows in analogously. 

\subparagraph*{Part 2 (reset):}

For this part, we have that $\unknownupdate =  \reset{u}$ for some $u \in \worddomain$. The update formula for the $\nextrelation$ relation under reset is defined as:
\[
\updateformula{\nextrelation}{\absreset}{u;x,y} \df \big( \nextrelation(x,y) \land ( (u < x) \lor (y < u) ) \big) \lor \big( \nextrelation(x,u) \land \nextrelation(u,y) \big).
\]

Looking at $\updateformula{\nextrelation}{\absreset}{u;x,y}$, we can see that $(x,y) \in \nextrelation$ and $(x,y) \in \nextrelation'$ when $(u < x) \lor (y < u)$. If we assume that $(x,y) \in \nextrelation$, it follows that there doesn't exist some element $v \in \worddomain$ such that $x<v<y$ and $w(v) \neq \emptyword$. Therefore we have that $(u < x) \lor (y < u)$ can only be false if $u=x$ or $u=y$ since there cannot be another element between $x$ and $y$ which has a symbol. Therefore if we have that $(x,y) \in \nextrelation$ and $(x,y) \notin \nextrelation'$ it must be that the update is $\reset{x}$ or $\reset{y}$. This is the correct behavior since if $w'(x)=\emptyword$ or $w'(y)= \emptyword$ then $x \not\newnextsym y$. 

We also have that $(x,y) \notin \nextrelation$ and $(x,y) \in \nextrelation'$ when $\nextrelation(x,u) \land \nextrelation(u,y)$. We can see that $\nextrelation(x,u) \land \nextrelation(u,y)$ is the case only when $x \nextsym u \nextsym y$ and if we have that $\unknownupdate = \reset{u}$ then it follows that there doesn't exist any element $v \in \worddomain$ such that $x<v<y$ and $w(v) \neq \emptyword$, therefore $x \newnextsym y$. Therefore the update formula $\updateformula{\nextrelation}{\absins{\zeta}}{u;x,y}$ has the desired behavior.

The following is the update formula for $\firstrel$:
\begin{multline*}
\updateformula{\firstrel}{\absreset}{u;x} \df \big( \firstrel(x) \land (u > x) \big) \lor \big( \firstrel(u) \land \nextrelation(u,x) \big) \lor \\ 
\big(\firstrel(u) \land \lastrel(u) \land (x \logeq \$) \big).
\end{multline*}

Looking at $\updateformula{\firstrel}{\absreset}{u;x}$, we can see that if $x \in \firstrel$ and $u>x$ then $x \in \firstrel'$. We can also see that if $u \in \firstrel$, i.e. we are setting $w'(u) = \emptyword$ where $\position{w}{u}=1$, then $x \in \firstrel'$ where $u \nextsym x$. This is because if $u \nextsym x$ then it follows that $\position{w}{x} = \position{w}{u}+1$ and therefore $\position{w}{x} = 2$ and because we are resetting $u$, $\position{w'}{x}=1$. 

We also have one edge case which is when $\firstrel(u)$ and $\lastrel(u)$. If this is the case, it follows that $|w| = 1$ and therefore $|w'| = 0$, i.e. $w' = \emptyword$. Therefore, we have that $\$ \in \firstrel$. We do this because given an insertion, of some element $v \in \worddomain$, it follows that $v< \$$ and therefore the update formula $\updateformula{\firstrel}{\absreset}{u;x}$ has the desired behavior.

The following is the update formula for $\lastrel$:
\begin{multline*}
\updateformula{\lastrel}{\absreset}{u;x} \df \big( \lastrel(x) \land (u<x) \big) \lor \big( \lastrel(u) \land \nextrelation(x,u) \big) \\
\lor \big( \firstrel(u) \land \lastrel(u) \land (x \logeq 1) \big).
\end{multline*}

The reasoning behind the update formula $\updateformula{\lastrel}{\absreset}{u;x}$ is analogous to the reasoning given earlier for the update formula $\updateformula{\firstrel}{\absreset}{u;x}$. 
\end{proof}

\subsection{Proof of \cref{lemma:eqsubstr}}
We first observe two results which help us in the actual proof of \cref{lemma:eqsubstr}:


\begin{lemma}
\label{obs:next}
If $y \nextsym z$ then $w[x,y] \cdot w[z,v] = w[x,v]$ where $x,y,z,v \in \worddomain$.
\end{lemma}
\begin{proof}
Because $y \nextsym z$ it follows that $w[y+1,z-1] = \emptyword$. Since we can write $w[x,v]$ as $w[x,y] \cdot w[y+1,z-1] \cdot w[z,v]$ and because $w[y+1,z-1] = \emptyword$, it follows that $w[x,y] \cdot w[z,v] = w[x,v]$.
\end{proof}

\begin{lemma}
\label{obs:insert}
If $w[x_1,y_1] = w[x_2,y_2]$ and we perform $\ins{\zeta}{u}$ then $w'[x_1,y_1] \neq w'[x_2,y_2]$ if $x_1 < u < y_1$ or $x_2 < u < y_2$.
\end{lemma}
\begin{proof}
If $x_1 < u < y_1$ then it follows that $|w'[x_1,y_1]| = |w[x_1,y_1]| + 1$ but since $w'[x_2,y_2]| = |w[x_2,y_2]|$ it follows that $w'[x_1,y_1] \neq w'[x_2,y_2]$. The reasoning for when $x_2 < u < y_2$ is analogous.

There is the case when $w(u) \neq \emptyword$ and $x_1 < u < y_1$ or $x_2 < u < y_2$. However since we have the restriction that an update \emph{must} change the word (\ie we cannot perform $\ins{\zeta}{u}$ if $w(u) = \zeta$), this case is trivial.
\end{proof}

\subsubsection{Actual proof of \cref{lemma:eqsubstr}}
\newcommand{\symfix}{\mu_{sym}(\oneopen, \oneclose,\twoopen,\twoclose)}
\begin{proof}
In a similar fashion to the proof of Lemma \ref{lemma:next}, we split this proof into two parts. For both parts we assume that $\equalsubstr$ is correct for a word-structure in some state, then prove that the update formula $\updateformula{\equalsubstr}{\absins{\zeta}}{u; \oneopen,\oneclose,\twoopen,\twoclose}$ correctly updates $\equalsubstr$. We have that $\equalsubstr$ is initialized to be $\emptyset$. If our update formulas are all in \ucq, then the equal substring relation can be maintained in $\dyncq$.

\subparagraph*{Part 1 (insertion):}

For this part of the proof, we have $\unknownupdate = \ins{\zeta}{u}$. Let $\equalsubstr'$ denote the relation $\{ (\oneopen,\oneclose,\twoopen,\twoclose) \mid \programstate \models \updateformula{\equalsubstr}{\absins{\zeta}}{u; \oneopen,\oneclose,\twoopen,\twoclose} \}$. The update formula for $\equalsubstr$ is:
\begin{multline*}
\updateformula{\equalsubstr}{\absins{\zeta}}{u; \oneopen,\oneclose,\twoopen,\twoclose} \df  \bigvee\limits_{i=1}^{9} \big( \substrsubform{i} \big) \land (\oneclose < \twoopen) \land \symel{\oneopen} \\ 
\land \symel{\oneclose} \land \symel{\twoopen} \land \symel{\twoclose}.
\end{multline*}

We have that for $\updateformula{\equalsubstr}{\absins{\zeta}}{u; \oneopen,\oneclose,\twoopen,\twoclose}$ to evaluate to true, it must be that $(\oneclose < \twoopen)$ and $\symel{\oneopen}$, which is only true when $w(\oneopen) \neq \emptyword$. Similarly, it must be that $w(\oneclose)$, $w(\twoopen)$ and $w(\twoclose)$ are all not the empty word. This is per the definition of the equal substring relation. Therefore, it is enough to show that if $\substrsubform{i} = \true$ then $w'[\oneopen,\oneclose] = w'[\twoopen,\twoclose]$ since the other cases of the equal substring relation definition have been dealt with. Note that $\substrsubform{i}$ is different for each $\zeta \in \Sigma$. 

Let $ \oneopen,\oneclose,\twoopen,\twoclose \in \worddomain$ be elements of our domain such that $\oneopen \leq \oneclose < \twoopen \leq \twoclose$. We have four cases to consider:

\begin{description}
\item[Case 1.] $w[\oneopen, \oneclose] = w[\twoopen,\twoclose]$ and $w'[\oneopen, \oneclose] \neq w'[\twoopen,\twoclose]$:
\end{description}

From Lemma \ref{obs:insert}, we know that if $w[\oneopen,\oneclose] = w[\twoopen,\twoclose]$ and we perform $\ins{\zeta}{u}$ where $\oneopen < u < \oneclose$ or $\twoopen < u < \twoclose$ then $w'[\oneopen,\oneclose] \neq w'[\twoopen,\twoclose]$. Therefore if $(\oneopen,\oneclose,\twoopen,\twoclose) \in R_{eq}$ and $\oneopen < u < \oneclose$ or $\twoopen < u < \twoclose$ then $\equalsubstr$ should be updated by the update formula such that $(\oneopen,\oneclose,\twoopen,\twoclose) \notin \equalsubstr'$. We now define the $\substrsubform{1}$:
\[
\substrsubform{1} \df \equalsubstr(\oneopen,\oneclose,\twoopen,\twoclose) \land \Big( (u<\oneopen) \lor \big( ( \oneclose < u ) \land ( u < \twoopen ) \big) \lor ( \twoclose < u ) \Big).
\]

If $\oneopen \leq u \leq \oneclose$ then $(u<\oneopen) = \false$, $(\oneclose < u) = \false$ and $(y_c < u) = \false$. Therefore we can see that $\substrsubform{1}$ will evaluate to false. If $\twoopen \leq u \leq \twoclose$ then $(u<\oneopen)= \false$, $(u<\twoopen)= \false$ and $(u<\twoclose) = \false$ and therefore $\substrsubform{1}$ evaluates to false. Hence, if $(\oneopen, \oneclose, \twoopen, \twoclose) \in \equalsubstr$ then it cannot be the case that $\oneopen<u<\oneclose$ nor can it be the case that $\twoopen<u<\twoclose$ for $(\oneopen, \oneclose, \twoopen, \twoclose) \in \equalsubstr'$. Indeed, it could be that $u=\oneopen$ and $w(\oneopen)=\zeta$ and therefore $w = w'$ even though $\substrsubform{1} = \false$, but this is dealt with using $\substrsubform{2}$, which we define later. Similar issues arise when $u=\oneclose$, $u=\twoopen$ and when $u=\twoclose$, but similarly they are all dealt with later on. Therefore, it can be seen that $\substrsubform{1}$ correctly maintains $\equalsubstr$ for this case.

\begin{description}
\item[Case 2.] $w[\oneopen, \oneclose] = w[\twoopen,\twoclose]$ and $w'[\oneopen, \oneclose] = w'[\twoopen,\twoclose]$:
\end{description}

This case is also dealt with by $\substrsubform{1}$. We again refer to Lemma \ref{obs:insert}. From this lemma, we know that if $w[\oneopen,\oneclose] = w[\twoopen,\twoclose]$ and we perform $\ins{\zeta}{u}$ but it is not the case that $\oneopen \leq u \leq \oneclose$ or $\twoopen \leq u \leq \twoclose$, then $w'[\oneopen,\oneclose] = w'[\twoopen,\twoclose]$. If it is not the case that $\oneopen \leq u \leq \oneclose$ or $\twoopen \leq u \leq \twoclose$, then we can see that $u<\oneopen \lor \big( \oneclose < u \land u< \twoopen \big) \lor \twoclose$ is true, and therefore if $(\oneopen,\oneclose,\twoopen,\twoclose) \in \equalsubstr$ then $\substrsubform{1}$ evaluates to true. It follows that $(\oneopen,\oneclose,\twoopen,\twoclose) \in \equalsubstr'$, which is the correct behavior in this case.
 
\begin{description}
\item[Case 3.] $w[\oneopen, \oneclose] \neq w[\twoopen,\twoclose]$ and $w'[\oneopen, \oneclose] = w'[\twoopen,\twoclose]$:
\end{description}

We have eight cases within Case 3, each case has an associated subformula. Since the subformulas are joined by disjunction to form $\updateformula{\equalsubstr}{\absins{\zeta}}{u; \oneopen,\oneclose,\twoopen,\twoclose}$, if one of the subformulas evaluates to true then $(\oneopen,\oneclose,\twoopen,\twoclose) \in \equalsubstr'$. Since we are in the case where  $w'[\oneopen, \oneclose] = w'[\twoopen,\twoclose]$, we wish to prove that $(\oneopen,\oneclose,\twoopen,\twoclose) \in \equalsubstr'$.

\begin{description}
\item[Case 3.1.] $u= \oneopen$ and $|w'[\oneopen, \oneclose]| > 1$:
\end{description}

For this case, we define $\substrsubform{2}$:
\begin{multline*}
\substrsubform{2} \df \exists v_1 \exists v_2 \colon \big(\equalsubstr(v_1, \oneclose, v_2, \twoclose) \land \nextrelation'(\oneopen, v_1)  \\
 \land \nextrelation'(\twoopen, v_2) \land \symbolrel{\zeta}(\twoopen) \land (u \logeq \oneopen) \big).
\end{multline*}

We can see that $\substrsubform{2}$ states that $(\oneopen,\oneclose,\twoopen,\twoclose) \in \equalsubstr'$ if there exists $v_1, v_2 \in \worddomain$, such that:
\begin{itemize}
\item $\equalsubstr(v_1, \oneclose, v_2, \twoclose)$ - which if true, we know that $w[v_1, \oneclose] = w[v_2, \twoclose]$.
\item $\nextrelation'(\oneopen, v_1) \land \nextrelation'(\twoopen, v_2)$ - which if true, we know that $\oneopen \newnextsym v_1$ and $\twoopen \newnextsym v_2$.
\item $\symbolrel{\zeta}(\twoopen)$ - which if true, we know that $w'[\twoopen,\twoopen] = w'[u,u] = \zeta$.
\item $u = \oneopen$.
\end{itemize}

Assume $\substrsubform{2} = \true$, we now show that $w'[\oneopen,\oneclose] = w'[\twoopen,\twoclose]$ must hold. If we have that $\substrsubform{2} = \true$ then we know that $w[\twoopen,\twoopen] = w[u,u]$ and that $w[v_1, \oneclose] = w[v_2, \twoclose]$, therefore it follows that:
\[
w'[u,u] \cdot w[v_1, \oneclose] = w'[\twoopen,\twoopen] \cdot w[v_2, \twoclose]
\]
and since $u = \oneopen$
\[
w'[\oneopen,\oneopen] \cdot w[v_1, \oneclose] = w'[\twoopen,\twoopen] \cdot w[v_2, \twoclose].
\]

We also have that the only change to the word-structure is that $w'(u) = \zeta$. Therefore all substrings that do not contain $u$ remain unchanged. Therefore:
\[
w'[\oneopen,\oneopen] \cdot w'[v_1, \oneclose] = w'[\twoopen,\twoopen] \cdot w'[v_2, \twoclose].
\]

Since we also have that $\oneopen \newnextsym v_1$ and $\twoopen \newnextsym v_2$, we can use Lemma \ref{obs:next} which gives us that:
\[
w'[\oneopen, \oneclose] = w'[\oneopen,\oneopen] \cdot w'[v_1, \oneclose]  \text{ and } w'[\twoopen, \twoclose] = w'[\twoopen,\twoopen] \cdot w'[v_2, \twoclose].
\]

Therefore we have shown that if $\substrsubform{2} = \true$ then $w'[\oneopen, \oneclose]  = w'[\twoopen, \twoclose]$. 

\begin{description}
\item[Case 3.2.] $\oneopen < u < \oneclose$ and $|w'[\oneopen, \oneclose]| > 1$:
\end{description}

For this case, we define $\substrsubform{3}$:
\begin{multline*}
\substrsubform{3} \df \exists z_1, z_2, z_3, z_4, v \colon \big(\nextrelation'(z_1, u) \land \nextrelation'(u,z_2) \land \nextrelation'(z_3, v)   \\
\land \nextrelation'(v,z_4) \land \equalsubstr(\oneopen, z_1, \twoopen, z_3) \land \equalsubstr(z_2,\oneclose,z_4,\twoclose ) \land \symbolrel{\zeta}(v)  \big).
\end{multline*}

We can see that $\substrsubform{3}$ states that $(\oneopen,\oneclose,\twoopen,\twoclose) \in \equalsubstr'$ if there exists $z_1,z_2,z_3,z_4,v \in \worddomain$ such that:
\begin{itemize}
\item $\nextrelation'(z_1, u)$ - which if true, we know that $z_1 \newnextsym u$.
\item $\nextrelation'(u,z_2)$ - which if true, we know that $u \newnextsym z_2$.
\item $\nextrelation'(z_3, v)$ - which if true, we know that $v_3 \newnextsym v$.
\item $\nextrelation'(v,z_4)$ - which if true, we know that $v \newnextsym z_4$.
\item $\equalsubstr(\oneopen, z_1, \twoopen, z_3)$ - which if true, we know that $w[\oneopen,z_1] = w[\twoopen,z_3]$.
\item $\equalsubstr(z_2,\oneclose,z_4,\twoclose)$ - which if true, we know that $w[z_2,\oneclose] = w[z_4,\twoclose]$.
\item $\symbolrel{\zeta}(v)$ - which if true, we know that $w'[u,u] = w'[v,v]$.
\end{itemize}

Let $\substrsubform{3} =\true$, we know that $w[\oneopen,z_1] = w[\twoopen,z_3]$, $w[z_2,\oneclose] = w[z_4,\twoclose]$ and $w'[u,u] = w'[v,v]$. Therefore, we can write:
\[
w[\oneopen,z_1] \cdot w'[u,u] \cdot w[z_2,\oneclose] = w[\twoopen,z_3] \cdot w'[v,v] \cdot w[z_4,\twoclose].
\]

Since the only change to the word-structure is that $w(u)$ is now $\zeta$ where $\zeta \in \Sigma$, we know that all subwords of the word-structure that do not contain $u$ remain unchanged, therefore:
\[
w'[\oneopen,z_1] \cdot w'[u,u] \cdot w'[z_2,\oneclose] = w'[\twoopen,z_3] \cdot w'[v,v] \cdot w'[z_4,\twoclose].
\]

Since we are assuming that $\substrsubform{3}=\true$, we also have that $z_1 \newnextsym u$ and $u \newnextsym z_2$, therefore $w'[\oneopen,\oneclose] = w'[\oneopen,z_1] \cdot w'[u,u] \cdot w'[z_2,\oneclose]$ and similarly because $v_3 \newnextsym v$ and $v \newnextsym z_4$ we have that $w'[\twoopen,\twoclose] = w'[\twoopen,z_3] \cdot w'[v,v] \cdot w'[z_4,\twoclose]$. This all follows from Lemma \ref{obs:next}. We therefore can see that $w'[\oneopen,\oneclose]=w'[\twoopen,\twoclose]$.

\begin{description}
\item[Case 3.3.] $u = \oneclose$ and $|w'[\oneopen, \oneclose]| > 1$:
\end{description}

For this case, we define $\substrsubform{4}$:
\begin{multline*}
\substrsubform{4} \df \exists v_1 \exists v_2 \colon \big(  \equalsubstr(\oneopen,v_1,\twoopen,v_2) \land \nextrelation'(v_1,u)  \\
\land \nextrelation'(v_2,\twoclose) \land (u \logeq \oneclose) \land \symbolrel{\zeta}(\twoclose) \big).
\end{multline*}

We now show that if $\substrsubform{4} = \true$, then $w'[\oneopen,\oneclose] = w'[\twoopen,\twoclose]$. If $\substrsubform{4} = \true$, then there exists $v_1, v_2 \in \worddomain$ such that:
\begin{itemize}
\item $\equalsubstr(\oneopen,v_1,\twoopen,v_2)$ - which if true, we know that $w[\oneopen,v_1] = w[\twoopen,v_2]$.
\item $\nextrelation'(v_1,u)$ - which if true, we know that $v_1 \newnextsym u$.
\item $\nextrelation'(v_2,\twoclose)$ - which if true, we know that $v_2 \newnextsym \twoclose$.
\item $\symbolrel{\zeta}(\twoclose)$ - which if true, we know that $w'[u,u] = w'[\twoclose,\twoclose]$.
\item $u = \oneclose$.
\end{itemize}   

Since $w[\oneopen,v_1] = w[\twoopen,v_2]$ and $w'[u,u] = w'[\twoclose,\twoclose]$, we know that:
\[
w[\oneopen,v_1] \cdot w'[u,u]= w[\twoopen,v_2] \cdot w'[\twoclose,\twoclose].
\]

Also since the only difference between the word before the update and after the update is the changing of $w(u)$ to $\zeta$, we can write:
\[
w'[\oneopen,v_1] \cdot w'[u,u]= w'[\twoopen,v_2] \cdot w'[\twoclose,\twoclose].
\]

Moreover, from Lemma \ref{obs:next} we know that $v_1 \newnextsym u$ and that $v_2 \newnextsym \twoclose$, therefore we can write that $w'[\oneopen,\oneclose] =w'[\oneopen,v_1] \cdot w'[u,u]$ and that $w'[\twoopen,\twoclose] = w'[\twoopen,v_2] \cdot w'[\twoclose,\twoclose]$. Therefore we have shown that when $\substrsubform{5} = \true$, we have that $w'[\oneopen,\oneclose] = w'[\twoopen,\twoclose]$.

\begin{description}
\item[Case 3.4.] $w'[u,u] = w'[\twoopen, \twoclose]$ and $u < \twoopen$:
\end{description}

For this case, we define $\substrsubform{5}$:
\[
\substrsubform{5} \df (u \logeq \oneopen) \land (\oneopen \logeq \oneclose) \land (\twoopen \logeq \twoclose) \land \symbolrel{\zeta}(\twoopen).
\]

We assume that $\substrsubform{5} = \true$ and then show that $w'[\oneopen, \oneclose] = w'[\twoopen,\twoclose]$ must hold. Since $u = \oneopen$ and $\oneopen = \oneclose$ it follows that $w'[u,u] = w'[\oneopen,\oneclose]$. Furthermore since $\twoopen = \twoclose$ we have that $w'[\twoopen,\twoclose] = w'[\twoopen,\twoopen]$. Therefore the equality $w'[\oneopen, \oneclose] = w'[\twoopen,\twoclose]$ can be rewritten as $w'[u,u] = w'[\twoopen,\twoopen]$ and since $\symbolrel{\zeta}(\twoopen)$, we know that $w'[u,u] = w'[\twoopen,\twoopen]$ is in fact the case. Hence, if $\substrsubform{5} = \true$ then $w'[\oneopen, \oneclose] = w'[\twoopen,\twoclose]$.

There are four other cases, although they are symmetric to the cases 3.1 to 3.4, i.e. we have that $\twoopen \leq u \leq \twoclose$ rather than $\oneopen \leq u \leq \oneclose$. Due to the fact that the cases are symmetrical, we have omitted the remaining proofs for said cases.

\begin{description}
\item[Case 4.] $w[\oneopen, \oneclose] \neq w[\twoopen,\twoclose]$ and $w'[\oneopen, \oneclose] \neq w'[\twoopen,\twoclose]$:
\end{description}

For this case, since $w'[\oneopen, \oneclose] \neq w'[\twoopen,\twoclose]$ it must be that $\updateformula{\equalsubstr}{\absins{\zeta}}{u; \oneopen,\oneclose,\twoopen,\twoclose}$ evaluates to false. Since we have exhaustively looked at all the cases where $w'[\oneopen, \oneclose] = w'[\twoopen,\twoclose]$ and shown that $\updateformula{\equalsubstr}{\absins{\zeta}}{u; \oneopen,\oneclose,\twoopen,\twoclose}$ evaluates to true, if $w'[\oneopen, \oneclose] \neq w'[\twoopen,\twoclose]$ it must be that  $\updateformula{\equalsubstr}{\absins{\zeta}}{u; \oneopen,\oneclose,\twoopen,\twoclose}$ evaluates to false. 

\subparagraph*{Part 2 (reset):}

For this part, we have that $\unknownupdate = \reset{u}$.
\begin{multline*}
\updateformula{\equalsubstr}{\absreset}{u;\oneopen,\oneclose,\twoopen,\twoclose} \df \bigvee\limits_{i=6}^{8}\big( \substrsubform{i} \big) \land (\oneclose < \twoopen) \land \symel{\oneopen} \\ 
\land \symel{\oneclose} \land \symel{\twoopen} \land \symel{\twoclose}.
\end{multline*}

\begin{description}
\item[Case 1.] $w[\oneopen, \oneclose] = w[\twoopen,\twoclose]$:
\end{description}

For this case we define the subformula $\substrsubform{6}$
\[
\substrsubform{6} \df \equalsubstr(\oneopen,\oneclose,\twoopen,\twoclose) \land \Big( (u<\oneopen) \lor \big( (\oneclose < u) \land (u< \twoopen) \big) \lor (\twoclose < u) \Big).
\]

This subformula states that if $(\oneopen,\oneclose,\twoopen,\twoclose) \in \equalsubstr$ and $\oneopen \leq u \leq \oneclose$ or $\twoopen \leq u \leq \twoclose$ then $\substrsubform{6}=\false$. Whereas, if $(\oneopen,\oneclose,\twoopen,\twoclose) \in \equalsubstr$ and it is not the case that $\oneopen \leq u \leq \oneclose$ or $\twoopen \leq u \leq \twoclose$ then $\substrsubform{6} = \true$ and hence $(\oneopen,\oneclose,\twoopen,\twoclose) \in \equalsubstr'$. This is due to the fact that we can only reset one element at a time, $u$, and therefore since $\oneclose < \twoopen$, if $u$ is in either $[\oneopen,\oneclose]$ or $[\twoopen,\twoclose]$ then $w'[\oneopen,\oneclose] \neq w'[\twoopen,\twoclose]$ because exactly one of them has changed. If it is not the case that $u$ is in either $[\oneopen,\oneclose]$ or $[\twoopen,\twoclose]$, then $w[\oneopen,\oneclose] = w'[\oneopen,\oneclose] $ and $w[\twoopen,\twoclose] = w'[\twoopen,\twoclose]$.

\begin{description}
\item[Case 2.] $w[\oneopen, \oneclose] \neq w[\twoopen,\twoclose]$:
\end{description}

We have two cases to explore, when $\oneopen \leq u \leq \oneclose$ and when $\twoopen \leq u \leq \twoclose$. These cases are symmetrical and therefore we only explore the case where $\oneopen \leq u \leq \oneclose$. If neither of these conditions are met, then it follows that $w'[\oneopen, \oneclose] \neq w'[\twoopen,\twoclose]$. For this case, we define $\substrsubform{7}$:
\begin{multline*}
\substrsubform{7} \df \exists z_1, z_2, z_3, z_4 \colon \big( \equalsubstr(\oneopen,z_1,\twoopen, z_3) \land \equalsubstr(z_2,x_c,z_4,y_c) \land \\
 \nextrelation(z_1,u) \land \nextrelation(u,z_2) \land \nextrelation(z_3,z_4) \big).
\end{multline*}

We assume that $\substrsubform{7}$ evaluates to true and show that $w'[\oneopen, \oneclose] = w'[\twoopen,\twoclose]$ must hold. If $\substrsubform{7}=\true$, then there must exists $z_{1 \dots 4}$ and $v$ such that:
\begin{itemize}
\item $\equalsubstr(\oneopen,z_1,\twoopen, z_3)$ - therefore $w[\oneopen,z_1] = w[\twoopen,z_3]$.
\item $\equalsubstr(z_2,\oneclose,z_4,\twoclose)$ - therefore $w[z_2,\oneclose] = w[z_4,\twoclose]$.
\item $ \nextrelation(z_1,u)$ - therefore $z_1 \nextsym u$.
\item $ \nextrelation(u,z_2)$ - therefore $u \nextsym z_2$.
\item $\nextrelation(z_3,z_4)$ - therefore $z_3\nextsym z_4$.
\end{itemize}
We can see that if $\substrsubform{7}$ holds that $z_1 \nextsym u \nextsym z_2$ but since $\unknownupdate = \reset{u}$ it follows that $z_1 \newnextsym z_2$. Therefore it follows that $w'[\oneopen,\oneclose] = w'[\oneopen,z_1] \cdot w'[z_2,\oneclose]$ and $w'[\twoopen,\twoclose] = w'[\twoopen,z_3]\cdot w'[z_4,\twoclose]$. Hence we can see that $w'[\oneopen,\oneclose] = w'[\twoopen,\twoclose]$.

We also have $\substrsubform{8}$ which is equivalent to $\substrsubform{7}$ but for the case where $\twoopen < u < \twoclose$ rather than $\oneopen < u < \oneclose$. We have omitted this due to the fact that it is analogous to $\substrsubform{7}$.
\end{proof}








\end{document}